\documentclass[review=false, screen=false, nonacm]{acmart}

\usepackage{booktabs} 

\usepackage[ruled]{algorithm2e} 

\usepackage[normalem]{ulem}
\usepackage{enumitem}
\usepackage{cancel}

\SetAlFnt{\small}
\SetAlCapFnt{\small}
\SetAlCapNameFnt{\small}
\SetAlCapHSkip{0pt}
\IncMargin{-\parindent}
\newcommand{\E}{\mathbb{E}}
\newcommand{\R}{\mathbb{R}}

\newtheorem*{remark}{Remark}

\usepackage{subcaption}
\usepackage[normalem]{ulem}
\acmJournal{TWEB}
\acmVolume{?}
\acmNumber{?}
\acmArticle{?}
\acmYear{2018}
\acmMonth{5}
\copyrightyear{?}

\setcopyright{acmlicensed}

\begin{document}

\author{Tim Hellemans}
\author{Benny Van Houdt}
\affiliation{%
  \institution{University Of Antwerp}
  \streetaddress{Middelheimlaan 1}
  \city{Antwerp}
  \postcode{2000}
  \country{Belgium}}
  
  \title[]{Improved Load Balancing in Large Scale Systems using Attained Service Time Reporting}

\begin{abstract}
Our interest lies in load balancing jobs in large scale systems consisting of multiple dispatchers and FCFS servers. In the absence of any information on job sizes, dispatchers 
typically use queue length information reported by the servers to assign incoming jobs.
When job sizes are highly variable, using only queue length information is clearly suboptimal and performance can be improved if some indication can be provided to the dispatcher
about the size of an ongoing job. In a FCFS server measuring the attained service time of the ongoing job is easy and servers can therefore report this attained service time together with
the  queue length when queried by a dispatcher.

In this paper we propose and analyse a variety of load balancing policies that exploit
both the queue length and attained service time to assign jobs, as well as policies 
for which only the attained service time of the job in service is used.
We present a unified analysis for all these policies in a large scale system under the usual
asymptotic independence assumptions.  The accuracy of the proposed analysis is illustrated using simulation.

We present extensive numerical experiments which clearly indicate that a significant improvement in waiting (and thus also in response) time may be achieved by using the attained service time information on top of the queue length of a server. Moreover, the policies which do not make use of the queue length still provide an improved waiting time for moderately loaded systems.
\end{abstract}

\maketitle

\section{Introduction}
Load balancing plays a vital role to achieve a low latency in large scale clusters. It was proven in \cite{winston77} that Join the Shortest Queue (JSQ) is the optimal policy to distribute jobs in a system with $N$ identical servers and exponential job sizes if jobs must be assigned immediately and
no jockeying between servers is allowed. Ever since, the JSQ policy has often been referred to as the golden standard for load balancers. As $N$ grows large, the overhead created by probing \textit{all} servers at every arrival becomes too large. Therefore the Shortest Queue $d$ (SQ($d$)) policy was introduced and analysed in various papers \cite{mitzenmacher2001power, aghajani2018pde, vasantam2018mean, vvedenskaya3}. A survey of recent advances can be found in \cite{van2018scalable}. However, the workload in many real systems consists of a mix of many short jobs
and some large jobs, where the large jobs contribute a significant part of the total workload (see e.g.~\cite{Sparrow,delgado2016job, delgado2015hawk}). When all servers have the same queue length, the JSQ policy simply assigns the incoming job arbitrarily, while it is better to assign the job to a server which is less likely to be serving a large job. Moreover, selecting a server with one long job may be worse than selecting a server with multiple small jobs. 

Recently the Least Loaded $d$ policy (LL($d$)) was analysed in \cite{hellemans2018power, ayesta2018unifying}, the LL($d$) policy assigns incoming jobs to a server which has the least amount of work left
among $d$ randomly selected servers. This allows one to improve the SQ($d$) policy significantly. The drawback of this policy is however that it only works if job size information
is available (which is mostly not the case) or when using a mechanism like late-binding (which brings significant overhead as jobs are not assigned immediately). While a server is typically unaware of its remaining workload, it can measure (up to some accuracy $\Delta$) the time it has spent processing the job(s) in service. This is especially
true in FCFS servers. As jobs that have been in service for a substantial amount of time
 are highly likely to be long jobs that potentially have a long residual service time, using the attained service time to assign jobs to servers
may improve performance.  

In this paper we propose a collection of load balancing policies that exploit both queue length
and attained service time (up to some granularity $\Delta$) information reported by the servers to assign jobs. Note that such policies do not require any knowledge of the size of incoming jobs or the job size distribution. We develop a unified analysis which may be used to analyse these load balancing policies in a large scale system under the usual asymptotic independence assumptions. Our main observation is that for small to moderate system loads and $d$ sufficiently large, the improvement from using the attained service time information is substantial, while this improvement decreases as the system becomes critically loaded. For example, even policies which solely rely on the attained service time of the job in service may outperform the SQ($d$) policy with $d=5$ for a large range of arrival rates (see Figure \ref{fig3a}).

As we may not aspire to approach the performance of the LL($d$) policy (as it is impossible to predict the exact workload using only the attained service time of one job and the queue length), we set ourselves a different goal. We define the Least Expected Workload policy (LEW($d$)) as the policy which assigns any incoming job to the server which has the least expected work left at its queue among $d$ randomly selected servers. Note that the LEW($d$) policy uses knowledge of the job size distribution to estimate the residual service time. We find that many of our policies that do not require such knowledge are able to achieve performance which is similar to that of the LEW($d$) policy.

When using (an indication of) the job sizes, there are multiple things one can do to improve waiting times. The Size Interval Task Assignment policy (also known as the SITA policy, see e.g.~\cite{harchol1999choosing, bachmat2010analysis}) distributes incoming jobs based on their size. In order to implement this policy, one needs to know the size of an incoming job. An alternative policy which does not require this information  is the Task Assignment based on Guessing Size policy (also known as the TAGS policy, see e.g.~\cite{harchol2000task, bachmat2020analysis}). For this policy, one sets cutoffs (which depend on the job size distribution) and a job migrates between servers when its service time exceeds these cutoffs. The policies considered in this paper do not require knowledge of the job size distribution and jobs do not migrate between servers after being assigned to a server by the dispatcher.

In our analysis we assume incoming jobs are Phase Type Distributed (further denoted by PH distributed). PH distributions are distributions with a modulating finite state background Markov chain \cite{latouche1} and any general positive-valued distribution can be approximated arbitrary close with a PH distribution. Further, various fitting tools are available online 
for PH distributions (e.g., \cite{panchenko1,Kriege2014}).

The main contributions of this paper are as follows:
\begin{enumerate}
\item 
We propose various load balancing policies that exploit queue length and attained service
time information reported by the servers. We demonstrate that all of our policies achieve
a significant reduction for the average waiting time of a job compared to SQ$(d)$ under low to moderate workloads, with a performance that is often close to the LEW($d$) policy. For some policies the waiting time is reduced for all workloads.
\item We present a unified analysis which is applicable for all policies under consideration
(and other variations thereof). This analysis may be of independent interest as it provides
a means to assess the performance of an M/PH/1 queue with queue length and attained service
time (up to some granularity $\Delta$) dependent arrival rates.
\item We validate the accuracy of the asymptotic independence assumption used in the analysis
using simulation experiments. 
\end{enumerate}

The paper is structured as follows. In Section \ref{sec:model} we formally define the model of interest. In Section \ref{sec:examples} we present seven different policies which we study throughout the paper. In Section \ref{sec:analysis} we present the method used to analyse our
load balancing policies. In Section \ref{sec:finiteAccuracy} we verify our analysis by means of simulation of finite systems. Section \ref{sec:numExperiments} consists of extensive numerical experimentation investigating the impact of all parameters in our proposed model. Finally, we conclude in Section \ref{sec:conclusion}.

\section{Model Description}\label{sec:model}
The model we consider consists of $N$ homogeneous servers (with $N$ large) which all process jobs using FCFS scheduling. We assume jobs arrive according to a Poisson $\lambda N$ process, these arrivals may occur to multiple dispatchers. We assume job sizes have a Phase Type 
distribution with representation $(\alpha, A)$. We use the notation $\mu = -A \textbf{1}$,
where $\textbf{1}$ is a column vector with all its entries equal to one, and denote by $m$ the number of phases of the job size distribution. Hence, the probability
of having a job size smaller than or equal to $x$ is given by $1- \alpha e^{Ax} \textbf{1}$.
Furthermore, we assume w.l.o.g.~that the mean job size is equal to one.

We pick a $\Delta>0$ and $r> 0$ and define $c_k=k\cdot \Delta$, for $1 \leq k\leq r$, and set
$c_0 = 0$ and $c_{r+1} = \infty$. We say a job is in layer$-k$ if its attained service time satisfies $a \in (c_{k-1}, c_{k}]$
and in layer$-0$ when the server is idle. 
When a server is queried by a dispatcher for queue length and attained service time information,
the server reports its queue length $\ell$ and the layer $k$ in which the attained service time of the job in process lies. This corresponds to stating that the servers measure the attained service time up to some granularity $\Delta$. Whenever an arrival occurs to a dispatcher, the dispatcher picks $d$ servers at random and queries these $d$ servers. The dispatcher then uses some policy based on the ($k,\ell$) values reported by the $d$ servers to assign the incoming job to one of these servers. We note that our analysis approach actually applies for any
set of threshold values $c_k$ such that $0=c_0 < c_1 < \dots < c_r < c_{r+1} = \infty$.

\section{Load Balancing Policies} \label{sec:examples}
In order to define our policies, we first define $\mathcal{R}_n=[0,\infty)^n$ with the lexicographic order. All our policies are based on the same basic idea, from every server, the dispatcher receives the layer and queue length information coded as $(k, \ell) \in \mathbb{N} \times \mathbb{N}$. The dispatcher maps the information $(k,\ell)$ to some value $\xi(k, \ell) \in \mathcal{R}_n$ which is interpreted as a measure for the \textquotedblleft aversion\textquotedblright\ of the chosen server. The incoming job is then assigned to the server for which the $\xi(k,\ell)$ 
value is the smallest (amongst the $d$ chosen servers), with ties being broken uniformly at random.
\begin{example}
Some basic examples of policies are random routing, which corresponds to picking $\xi(k,\ell)=0$ for all $(k,\ell)$ and SQ($d$) for which one sets $\xi(k,\ell) = \ell$.
\end{example}
We now introduce a number of load balancing policies that are all described by defining $\xi:\mathbb{N} \times \mathbb{N} \rightarrow \mathcal{R}_n$. For our policies, we always have $\xi(0,0)=(0)_{i=1}^n$ and $\xi(k,\ell) \neq (0)_{i=1}^n$ if $\ell \geq 1$. This ensures that we always assign to idle queues if possible. As jobs with larger attained service times are more
likely to be large jobs with a potentially large residual service time, $\xi$ is always chosen to be non-decreasing in $k$.

Throughout, we assume that $\xi$ is chosen such that our model remains stable for $\lambda < 1$. More often than not, it suffices to note that the load balancing policy outperforms the random routing policy.

\subsection{SQ($d$) with Runtime based Tie Breaking (SQ($d$)-RTB)}\label{sec:RTB}
This policy mainly relies on the queue length information, but in case multiple chosen servers have the same number of pending jobs, the job is routed to the server for which the job at the head of the queue has currently received the least service, that is, is in the lowest layer $k$. The intuition is that the job in service is more likely to be a short job that will finish soon.
For this policy, we set:
\begin{equation} \label{eq:xi_SQ_RTB}
\xi(k, \ell) = (\ell, k).
\end{equation}
We further refer to this policy as the SQ($d$) with Runtime based Tie Breaking policy, denoted by SQ($d$)-RTB. It is similar to SQ($d$) but may improve performance by using the attained service time to resolve ties.

\subsection{SQ($d$) with Runtime Exclusion (SQ($d$)-RE($T$))}\label{sec:RE}
For this policy, we only rely on the queue length information, as long as the attained service
time does not exceed some threshold $T$ (e.g. $T = 2$). When a server is queried it simply replies by stating its queue length and whether or not the attained service time
is more than time $T$. This corresponds to setting $\Delta = T$ and $r=1$ in our model.
Whenever there are $1\leq d' \leq d$ servers for which the attained service time does not exceed the threshold $T$, we assign the job to the server with the least number of jobs amongst the $d'$ chosen queues. If all $d$ servers report an attained service time above $T$ (meaning $k = 2$), the job is routed to the server with the least number of jobs amongst all $d$ chosen servers. The idea is that we assume a job is \textit{large} when its runtime \textit{significantly exceeds} the average runtime of a job. For this policy, we define:
\begin{equation}\label{eq:xi_SQ_RE}
\xi(k,\ell)=(k, \ell),
\end{equation}
We refer to this policy as  SQ($d$) with Runtime Exclusion, denoted by SQ($d$)-RE($T$).
  In \cite{mitzenmacher2020queues}, a somewhat similar policy is studied, where one uses the SQ($d$) policy to select a queue and each job carries one bit of information (to be interpreted as an indication of whether the job is large). Long jobs are put at the back of the queue, while short jobs are put at the head of the queue. The main difference lies in the fact that we infer the job size information at runtime and do not assume jobs carry an indication about their size.
Further all jobs are placed at the back of the queue in our case. 

\subsection{SQ($d$)-RTB with Runtime Exclusion (SQ($d$)-RTB-RE($T$))}
This policy is a mix of SQ($d$)-RTB and SQ($d$)-RE($T$). We set some threshold $T > 0$ where we suspect that a job is large once the attained service time exceeds this threshold. When $d' \geq 1$ of the randomly chosen servers report an attained service time smaller than $T$, we employ the SQ($d$)-RTB policy to decide which of these $d'$ servers receives the incoming job. When the attained service time of the $d$ selected servers exceeds $T$, we use SQ($d$)-RTB to pick a server. This policy can also be described by defining an appropriate $\xi$:
\begin{equation}\label{eq:xi_SQd_RTB_RE}
\xi(k, \ell)
=
(\delta\{k > s\}, \ell, k),
\end{equation}
here $\delta\{A\}$ is equal to one if $A$ is true and zero otherwise and $T = \Delta s = c_s$ for some $s \leq r$.
We refer to this policy as the SQ($d$)-RTB with Runtime Exclusion policy, denoted by SQ($d$)-RTB-RE($T$). 

\subsection{Least Attained Service (LAS($d$))}
For this policy, we assume that the dispatcher assigns incoming jobs to the server for which the job at the head of the queue has attained the least amount of service among $d$ randomly
selected servers. This policy is defined by:
\begin{equation}\label{eq:xi_LAS}
\xi(k, \ell) = k.
\end{equation}
We further refer to this policy as the Least Attained Service policy, denoted as LAS($d$). The success of this policy relies on having highly variable jobs, such that it is more important to know whether the job at the head of a queue is large rather than knowing the number of jobs in the queue. Note that if we pick $\Delta > 0$ sufficiently small, the probability of having a tie tends to zero.

\subsection{LAS($d$) with Queue length based Tie Breaking (LAS($d$)-QTB)} \label{sec:LASQTB}
For this policy, we assign the job to the queue for which the attained service time of the job at the head of the queue is minimal, but in case there is a tie between multiple servers, we assign the incoming job to the server with the fewest number of waiting jobs. To this end, we define:
\begin{equation}\label{eq:LAS_QTB}
\xi(k,\ell)=(k, \ell).
\end{equation}
We refer to this policy as LAS($d$) with Queue length based Tie Breaking, which we denote by LAS($d$)-QTB.
\begin{remark}
While for other policies, having a small $\Delta > 0$ is always beneficial, the performance of LAS($d$)-QTB may actually improve by having a larger value of $\Delta$. In particular, letting $r=1$ and $c_1=T$ this policy reduces to SQ($d$)-RE($T$). While setting an arbitrary $r$ and $c_1,\dots,c_r$, this policy may be viewed as an SQ($d$) policy with multiple thresholds.
\end{remark}

\subsection{Runtime Exclusion (RE($d, T$))}
For this policy, we set $r=1$ and $\Delta=T$. In this case having $k=0$ means that the
server is idle, $k=1$ means the job in service has an attained service time below $T$ and
$k=2$ otherwise. 
 We find:
\begin{equation}\label{eq:xi_RE}
\xi(k,\ell)=k.
\end{equation}
We refer to this policy as the Runtime Exclusion policy, denoted by RE($d, T$).
It is a special case of the LAS($d$) policy with $r=1$.

\subsection{Least Expected Workload (LEW($d$))}
For this policy we assume that the job size distribution of the incoming jobs is known, such that
we can compute the mean residual service time given the attained service time. In this case we can use the more refined information of the expected residual service time $\E[X \mid X > c_k] - c_k$ to decide which queue should receive the incoming job. This policy also fits in our framework by defining:
\begin{equation}\label{eq:xi_EW}
\xi(k, \ell)=(\ell-1) \cdot \E[X] + \E[X \mid X \geq c_k] - c_k,
\end{equation}
we refer to this policy as the Least Expected Workload policy, denoted by LEW($d$).
\begin{remark}
As our main objective is to study load balancers which are not aware of the job size distribution, we only use this policy to see how it compares with the other strategies. This policy may be viewed as an idealized version of what the other policies attempt to do: avoid queues with a large expected workload. As such, we view the performance of LEW($d$) as the goal of what the other policies try to achieve. In Section \ref{sec:numExperiments} we find that the performance of some policies indeed closely approximates that of LEW($d$).
\end{remark}

\section{Model Analysis} \label{sec:analysis}

In this section we use the cavity approach presented in \cite{bramsonLB} to study the 
performance of the load balancing algorithms introduced in the previous section. We refer to the attained service time of a job
at a server as its {\it age} $a$. Note that in our case the age $a$ of a job does not include the
waiting time of the job, only the time it has been in service.

\subsection{Description of the Cavity Map} \label{sec:cavity}
The cavity process intends to capture the evolution of a single server assuming
asymptotic independence among servers (see below). The state of a single server is captured
by the service phase ($j$), the age ($a$) of the job in service and the queue length ($\ell$). The state of a server is thus denoted as a triple ($j,a,\ell$). As arrivals occur according to a Poisson($\lambda N$) process and each arrival has a probability of $\frac{d}{N}$ to select any particular queue, the rate at which a server is selected as one of the $d$ random servers is equal to $\lambda d$, we refer to this rate as the {\it potential} arrival rate. At each potential arrival, $d-1$ independent copies of the queue at the cavity are considered. The state 
$(j,a,\ell)$ for each of these $d-1$ 
independent copies at time $t$ has the same distribution as the distribution of the
queue at the cavity at time $t$. The potential arrival is assigned to one of the $d$ selected queues based on the $(k,\ell)$ values reported by the 
queue at the cavity and the $d-1$ independent queues.  
Thus, the \textit{actual} arrival rate depends on both the queue length $\ell$ and the 
layer $k$ containing the age $a$ of the job at the head of the cavity queue at the time of a potential arrival. 
Let us denote this value by $\lambda_{act}(k,\ell)$. In general, we find that the number of jobs present in the queue at the cavity increases by one with a rate equal to $\lambda_{act}(k,\ell)$, whilst its age $a$ continuously increases at rate one, and the service phase $j$ evolves as dictated by the phase type distribution $(\alpha, A)$. When a job completes service, the age $a$ jumps to
zero, $\ell$ decreases by one and a new job starts service if present. 

In order to formally prove that the results presented in this paper correspond to
the limiting behavior as the number of servers $N$ tends to infinity, the modularized program
presented in \cite{bramsonLB} can be followed:
\begin{itemize}
\item[\textbf{a.}] \textbf{Asymptotic Independence.} Demonstrate $\Pi^N \rightarrow \Pi$ as $N \rightarrow \infty$, where $\Pi^N$ is the stationary distribution for the studied policy with $N$ queues and $\Pi$ is a stationary and ergodic distribution on $(\{1,\ldots,m\}\times[0,\infty)\times \mathbb{N})^\infty$. Show that the limit $\Pi$ is unique. Show that, for every $k$:
$$
\Pi^{(k)} = \bigotimes_{i=1}^k \Pi^{(1)},
$$
where $\Pi^{(k)}$ is $\Pi$ restricted to its first $k$ coordinates.
\item[\textbf{b.}] \textbf{The queue at the cavity.} Let $A_{act}^N$ denote the process of actual arrivals to the first queue. Show that $A_{act}^N \rightarrow \lambda_{act}$ in distribution as the number of servers $N$ tends to infinity. 
\item[\textbf{c.}] \textbf{Calculations.} Given $\lambda_{act}$, the actual arrival rates, analyse the queue at the cavity in the large $N$ limit using queueing techniques to express $\Pi^{(1)}$ as a function of $\lambda_{act}$:
$$
\Pi^{(1)}=T(\lambda_{act}).
$$
Moreover, the arrival rate is also determined by the state of a server $\Pi^{(1)}$ we thus have:
$$
\lambda_{act}= H(\Pi^{(1)}).
$$
We then must determine the fixed point of the cavity map, that is, solve
the equation $\Pi^{(1)}=T(H(\Pi^{(1)}))$  to obtain $\Pi^{(1)}$.
\end{itemize}
In this work, we focus on $\textbf{c}$, the computational step of the program. We present a numerical method to compute $\Pi^{(1)}$ for the load balancing policies in Section \ref{sec:examples}. For ease of notation we denote $\Pi^{(1)}$ as $\pi$.
We validate our results using simulation to show that the obtained solutions indeed correspond to the system under study (as $N\rightarrow \infty$).

\subsection{Obtaining the Steady State} \label{sec:obtain_stationary}

In this section we indicate how to compute $\pi$ given $\lambda_{act}$.
Given the discussion in the previous subsection, this step corresponds to determining the
steady state of a 
queueing system with the following characteristics:
\begin{enumerate}
\item A single server queue that serves jobs in FCFS manner.
\item Service times of a job follow an order $m$ 
phase type distribution with parameters $(\alpha,A)$.
\item Poisson arrivals with a rate that depends on the queue length $\ell$ and
the attained service time $a$ (if the queue is busy).  
\item No arrivals when the queue length equals some large value $B$.
\end{enumerate}
More specifically, let $0 = c_0 < c_1 < \ldots < c_r < c_{r+1} = \infty$, then  
the dependence on $a$ is such that the Poisson arrival rate is only influenced by 
$k$, where $k$ is the unique index such that $a \in (c_{k-1},c_k]$.
In other words, the queueing system under consideration is fully determined by
$B$, $(\alpha,A)$ and a set of arrival rates $\{\lambda_0\} \cup \{ \lambda_{k,\ell} |
k \in \{1,\ldots,r+1\}, \ell \in \{1,\ldots,B-1\}\}$. The reason why we can assume
an arrival rate equal to zero when the queue length equals some $B < \infty$ is explained
further on. 

For the purpose of determining a fixed point of the cavity map, it suffices to 
develop a method to compute
the steady-state probabilities $\pi^{busy}_{k,\ell,j}$ that the 
service phase equals $j$, queue length equals $\ell$ and the attained service time $a$ 
belongs to $(c_{k-1},c_k]$ {\it given that the queue is busy}, for $k=1,\ldots,r+1$, $j=1,\ldots,m$
and $\ell =1,\ldots,B$. These probabilities are not affected by $\lambda_0$.

We start by defining a finite state discrete-time Markov chain on the state space
\[\mathcal{S} = \{(k,\ell,j) | k = 0,\ldots,r; \ell = 1,\ldots,B; j = 1,\ldots,m\},\]
by observing the queue whenever the attained service time equals $c_k$ for some $k=0,\ldots,r$. Note that this implies that we observe
the queueing system exactly $i$ times for a job with length $z \in [c_{i-1},c_i)$:
once when the service starts and $i-1$ times during its service.

Given that this DTMC is in state $(k,\ell,j)$ we can have two types of transitions:
\begin{enumerate} 
\item The job in service remains in service for at least $c_{k+1}-c_k$ more time and the chain 
transitions to a state of the form $(k+1,\ell',j')$ with $\ell' \geq \ell$.
\item The length of the job in service is below $c_k$ and 
the chain transitions to a state of the form $(0,\ell',j')$ with $\ell' \geq \ell-1$.
\end{enumerate} 
Hence, if we order the states in $\mathcal{S}$ in lexicographical order,
the transition probability matrix $P^{DTMC}$ has the following form
\begin{equation}\label{eq:PDTMC}
P^{DTMC} = \begin{pmatrix}
\Xi^{(1)} & \Lambda^{(1)} &  & &  \\
\Xi^{(2)} &  & \Lambda^{(2)} & & \\
 \vdots &  & & \ddots &  \\
\Xi^{(r)} & & &  & \Lambda^{(r)}  \\
\Xi^{(r+1)} & & & &    \\
\end{pmatrix},
\end{equation}
where $\Xi^{(k)}$ and $\Lambda^{(k)}$, for $k=1,\ldots,r+1$, are
square matrices of size $mB$.
To express these matrices we define the size $mB$ matrix
\begin{equation}\label{eq:G}
G = (I_B \otimes \mu) \cdot \left( \begin{pmatrix}
1 & 0 & \dots & 0\\
1 & 0 & \dots & 0\\
0 & 1 & \dots & 0\\
0 & 0 & \ddots & 1
\end{pmatrix}
\otimes \alpha \right),
\end{equation}
where $\mu$ was defined as $-A \textbf{1}$ and the size $mB$ matrices
\begin{equation}\label{eq:Fk}
F^{(k)} = I_{B} \otimes A + \begin{pmatrix}
-\lambda_{k,1} & \lambda_{k,1} & & & & \\
& -\lambda_{k,2} & \lambda_{k,2} & & & \\
& & \ddots & \ddots & & \\
& & & -\lambda_{k, B-1} & \lambda_{k,B-1} &\\
& & & & -\lambda_{k, B} &
\end{pmatrix} \otimes I_m,
\end{equation}
for $k=1,\ldots,r+1$, with $\lambda_{k,B} = 0$. 

The matrix $G$ contains the rates at which 
$(\ell,j)$ changes due to a service completion, which does not depend on
the attained service time.
The matrix $F^{(k)}$ 
captures the evolution of the queue length and service phase 
when there is no service completion and 
the attained service time is between
$c_{k-1}$ and $c_k$. 
Based on these interpretations we have, for $k=1,\ldots,r+1$,
\begin{equation}\label{eq:Lambdak}
\Lambda^{(k)}=e^{F^{(k)} (c_k-c_{k-1})},
\end{equation} 
and $\Xi^{(k)} = \Psi^{(k)} G$, where
\begin{equation}\label{eq:Psik}
\Psi^{(k)} = \int_{0}^{c_k-c_{k-1}} e^{F^{(k)} \delta} \, d\delta = (I_{mB} - \Lambda^{(k)}) \cdot (-F^{(k)})^{-1}.
\end{equation}
Let $\hat \pi^{(k)}$ be the size $mB$ vector holding the steady state probabilities of
the DTMC characterized by $P^{DTMC}$ corresponding to the states of the
form $(k,\ell,j)$, then due to the structure of $P^{DTMC}$, we have
\[\hat \pi^{(k)} = \hat \pi^{(k-1)} \Lambda^{(k)} = \hat \pi^{(0)}  \prod_{i=1}^k \Lambda^{(i)},\]
for $k=1,\ldots,r+1$. Using the first balance equation, the vector $\hat \pi^{(0)} $
is therefore found as 
\[ \hat \pi^{(0)}  = \hat \pi^{(0)} \sum_{k=0}^{r} \left( \prod_{i=1}^k \Lambda^{(i)}\right) \Xi^{(k+1)}.\]
As $\Xi^{(k)} = \Psi^{(k)} G$, the above equation can be restated as $\hat \pi^{(0)}  = \hat \pi^{(0)} \Omega^{(1)}$, where
$\Omega^{(r+1)} = \Psi^{(r+1)} G$ and 
\[ \Omega^{(k)} =  \Psi^{(k)} G + \Lambda^{(k)} \Omega^{(k+1)}. \]
Hence, the time complexity to compute the steady state probabilities of the DTMC
characterized by $P^{DTMC}$ equals $O(r m^3 B^3)$.

As stated before, our aim is to compute the steady state probabilities  
 $\pi^{busy}_{k,\ell,j}$ that the service phase equals $j$, queue length equals $\ell$ and the attained service time $a$ belongs to $(c_{k-1},c_k]$ given that queue is busy. 
 These probabilities can now be computed from the steady state probabilities
$\hat \pi_{k,\ell,j}$ of the DTMC, by looking at the amount of time that the
queue length and service phase equal $(\ell',j')$ in  between
two points of observation of the DTMC. Note that
entry $((\ell,j),(\ell',j'))$ of the matrix $\Psi^{(k)}$ contains the 
expected amount of time that the queue length equals $\ell'$ and the server
is in phase $j'$ during a single transition of the DTMC from state
$(k-1,\ell,j)$ to any other state. We therefore have
\[ \pi^{busy}_{k,\ell',j'} = \nu ( \hat  \pi_{(k-1)} \Psi^{(k)})_{(\ell',j')},
\] 
for $k=1,\ldots,r+1$, $\ell' = 1,\ldots,B$ and $j'=1,\ldots,m$, with $\nu$ 
being the normalization constant. As each job brings $\E[X]=1$ work on average, the queue at
the cavity is empty with probability $1-\lambda$. This allows us to express the stationary distribution $\pi$ by setting $\pi_0=\lambda$ and renormalizing $\pi^{busy}$ to sum to $\lambda$.

\subsection{Determining the Arrival Rate} \label{sec:lambda}
In this section we indicate how to determine $\lambda_{act}$ given $\pi$.
For any $(k,\ell)$, we define:
\begin{equation}\label{eq:setAB}
\mathcal{A}(k,\ell) = \{(k', \ell') \mid \xi(k', \ell') \geq \xi(k,\ell) \} \qquad \mathcal{B}(k,\ell) = \{(k', \ell') \mid \xi(k', \ell') > \xi(k,\ell)\}.
\end{equation}
We denote $x_{k, \ell} = \sum_{j} \pi_{k,\ell,j}$ the probability that, in the stationary regime, the queue at the cavity is in layer-$k$ with queue length $\ell$. Furthermore we let:
\begin{equation}\label{eq:uvw}
u_{k,\ell} = \sum_{(k',\ell') \in \mathcal{A}(k,\ell)} x_{k', \ell'} \qquad v_{k,\ell} = \sum_{(k', \ell') \in \mathcal{B}(k,\ell)} x_{k',\ell'} \qquad w_{k,\ell} = u_{k,\ell} - v_{k,\ell}.
\end{equation}
\begin{remark}
When $\xi$ is injective, we simply have $w_{k,\ell} = x_{k,\ell}$.
\end{remark}
 We obtain the arrival rate in layer-$k$ with queue length $\ell$ in the following Proposition:
\begin{proposition}\label{prop:arr_rate_xi}
The arrival rate to the queue at the cavity, given its queue length $\ell$ and service layer-$k$ is given by:
\begin{equation}\label{eq:arr_rate_xi}
\lambda_{act}(k,\ell) = \frac{\lambda}{w_{k,\ell}} \left( u_{k, \ell}^d - v_{k, \ell}^d \right).
\end{equation}
\end{proposition}
\begin{proof}
Whenever a job arrives to the system, it samples the queue at the cavity and $d-1$ i.i.d.~servers with distribution $x_{k,\ell}$. Potential arrivals occur with rate $\lambda \cdot d$ and each potential arrival joins the queue at the cavity with probability $\frac{1}{j+1}$ if $j$ of the chosen servers are in some state $(k', \ell')$ with $\xi(k',\ell') = \xi(k, \ell)$, while all other servers are in state $(k', \ell')$ with $\xi(k', \ell') > \xi(k,\ell)$. We obtain:
\begin{align}
\lambda_{act}(k, \ell)
&= \lambda d \sum_{j=0}^{d-1} \frac{1}{j+1} \binom{d-1}{j} \cdot w_{k, \ell}^j \cdot v_{k, \ell} ^{d-1-j}\label{eq:arr_rate_xi_num_stable}\\
&= \frac{\lambda}{w_{k,\ell}}{((w_{k,\ell}+v_{k,\ell})^d - v_{k,\ell}^d)}, \nonumber
\end{align}
which simplifies to \eqref{eq:arr_rate_xi}, finishing the proof.
\end{proof}
\begin{remark}
While \eqref{eq:arr_rate_xi} is the more elegant formula, the expression in \eqref{eq:arr_rate_xi_num_stable} is actually more stable for numerical purposes. 
\end{remark}

\subsection{Iterative Procedure}\label{sec:iterative}
In this section we show how to compute the fixed point $\pi$ described in Section \ref{sec:cavity}  using Sections \ref{sec:obtain_stationary} and \ref{sec:lambda}. 
The procedure in Section \ref{sec:obtain_stationary} corresponds to the function $T$ such that the stationary distribution of the queue at the cavity $\pi = T(\lambda_{act})$. The 
result in Section \ref{sec:lambda} corresponds to the map $H$ such that $\lambda_{act} = H(\pi)$. The fixed point of the cavity map is given by the fixed point $\pi$ which satisfies $\pi = T(H(\pi))$. To this end, we propose the following iterative scheme to compute $\pi$ :
\begin{enumerate}
\item Pick some $\pi^{(0)}$ (e.g.~$\pi^{(0)}_0=1$), set some tolerance $tol$ and $n=0$.
\item Compute $\pi^{(n+1)}=T(H(\pi^{(n)}))$. \label{enumerate:it_step2}
\item If $\|\pi^{(n+1)} - \pi^{(n)}\|_1 = \sum_{k,\ell,j} | \pi^{(n+1)}_{k,\ell,j} - \pi^{(n)}_{k,\ell,j}| < tol$ we accept $\pi^{(n+1)}$, otherwise increment $n$ by one and return to step \ref{enumerate:it_step2}.
\end{enumerate}

Throughout our numerical experiments we typically employ the tolerance $tol=10^{-10}$. Moreover, as illustrated in Figure \ref{fig1}, $\|\pi^{(n+1)} - \pi^{(n)}\|_1$ decreases exponentially in $n$ and the number of iteration required is typically below $50$, where each iteration
requires $O(m^3B^3r)$ time. If we start with $\pi^{(0)}_0=1$, then setting $B=n$ during the
$n$-th iteration suffices. We can even make use of a smaller $B$ value as the arrival
rates $\lambda_{k,\ell}$ decrease very rapidly in $\ell$ for most policies considered.

\begin{figure*}[t]
\begin{subfigure}{0.43\textwidth}
\centering
\captionsetup{width=.8\linewidth}
\includegraphics[width=1\linewidth]{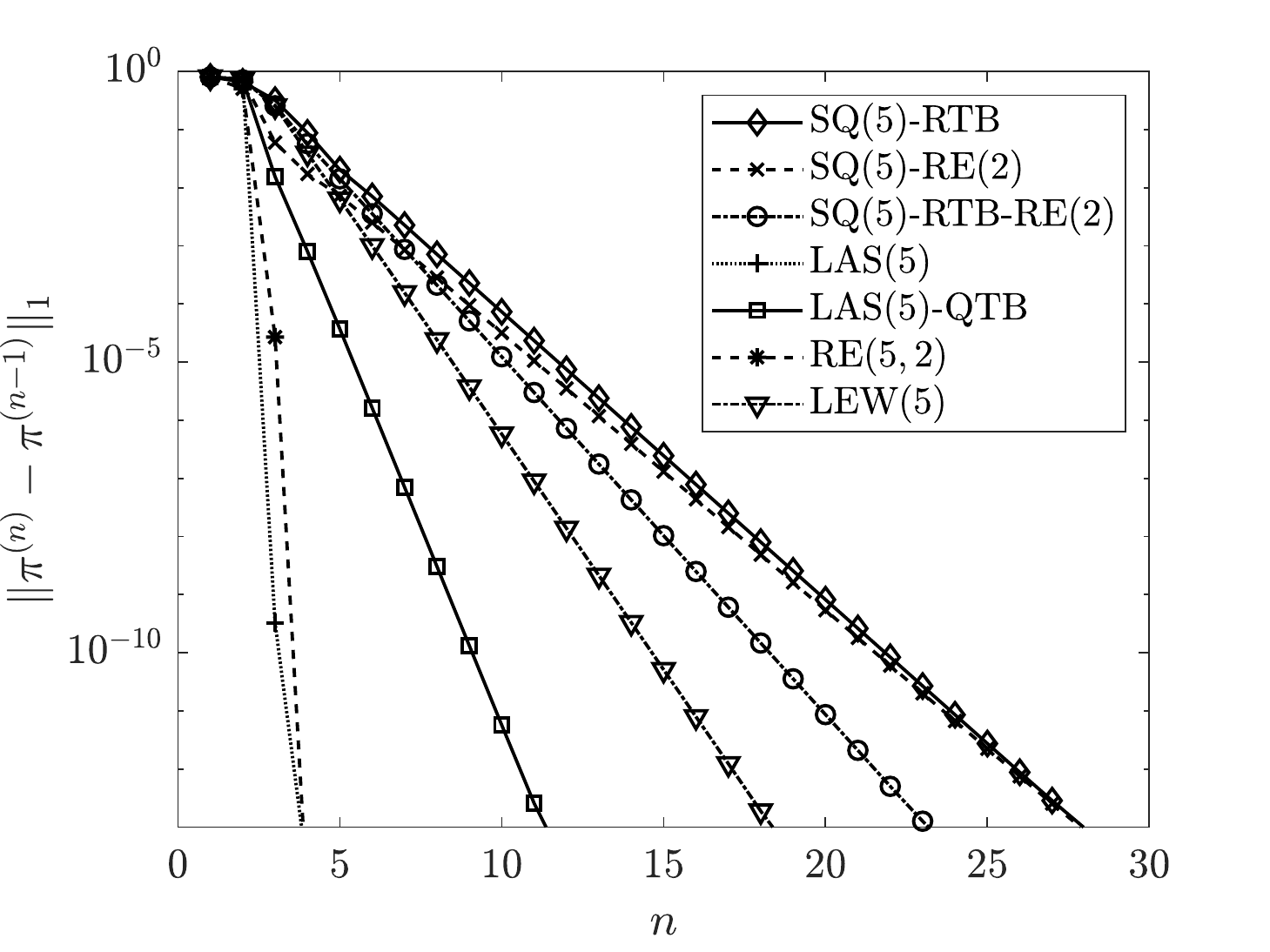}
\caption{We set $d=5$ and consider all policies in Section \ref{sec:examples}.}
\label{fig1a}
\end{subfigure}
\begin{subfigure}{.4\textwidth}
\centering
\captionsetup{width=.8\linewidth}
\includegraphics[width=1\linewidth]{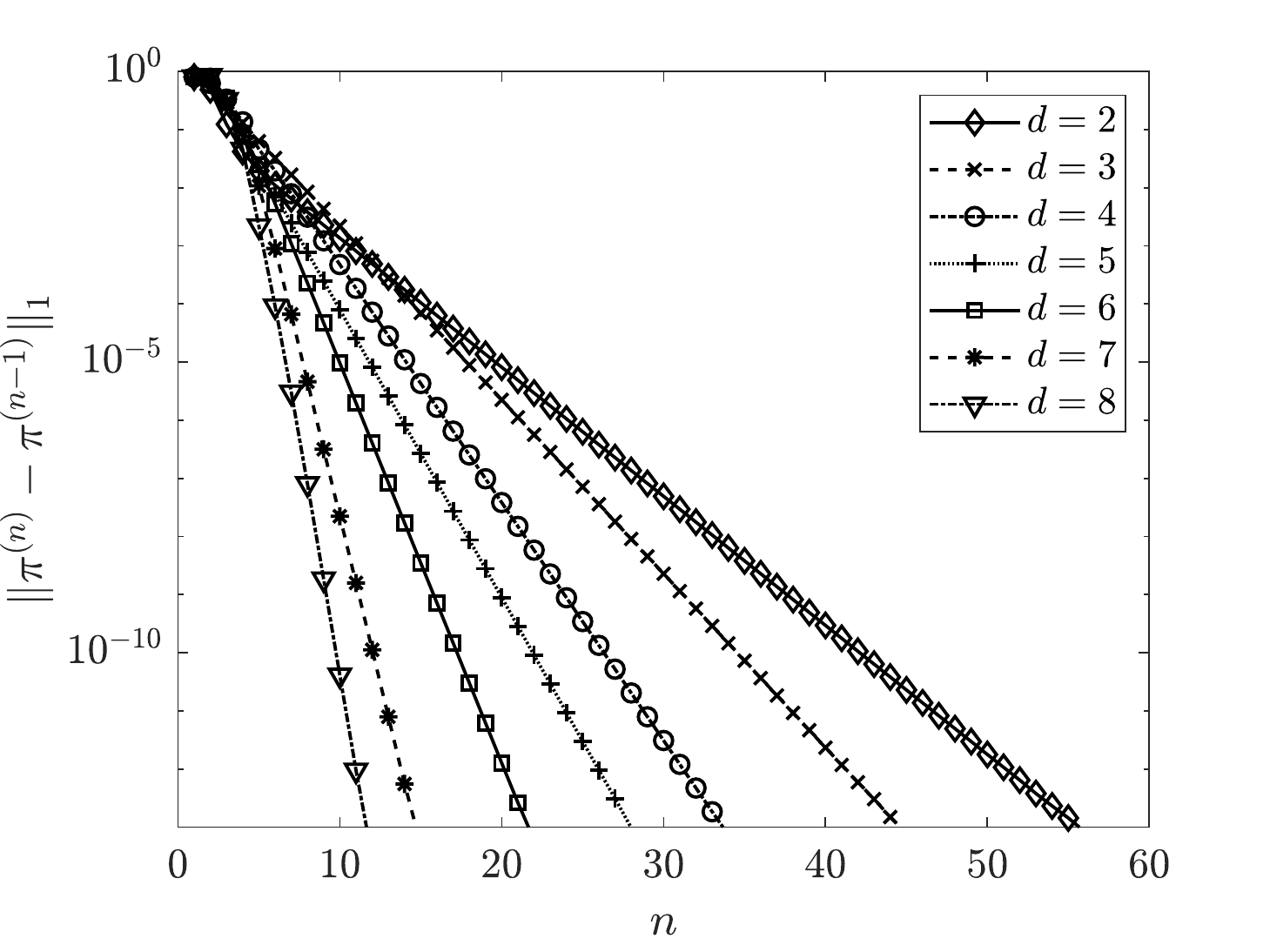}
\caption{We use the SQ($d$)-RTB for $d=2,\ldots,8$.}
\label{fig1b}
\end{subfigure}
\caption{Convergence of $\|\pi^{(n)} - \pi^{(n-1)}\|_1$ for $\lambda=0.8$, $\Delta=0.1$ and HEXP($10, 1/2$) job sizes.}
\label{fig1}
\end{figure*}

\subsection{Obtaining Performance Measures}
Given the stationary distribution $\pi$ obtained in Section \ref{sec:iterative}, we show how to obtain the performance measures associated to our model. In particular we are interested in the average queue length, response time and waiting time. The expected queue length can easily be computed as $\E[Q]=\sum_{k,\ell,j} \ell \cdot \pi_{k,\ell,j}$, while the expected response time can then be derived using Little's Law $\E[R]=\E[Q]/\lambda$. The expected waiting time is then given by $\E[W]=\E[R]-1$.

In order to compute the waiting time distribution, we denote by $J_{k,\ell,j}$ the probability that a random job arriving to the system joins a queue with length $\ell$ for which the job at the head of the queue is in phase $j$ and resides in layer-$k$. It is not hard to see (similar to the proof of Proposition \ref{prop:arr_rate_xi}) that:
\begin{equation} \label{eq:Jkellj}
J_{k, \ell, j} = d\cdot \pi_{k,\ell,j} \cdot \sum_{s=0}^{d-1} \frac{1}{s+1} \binom{d-1}{s} w_{k,\ell}^{s} \cdot v_{k, \ell}^{d-1-s} = \frac{\pi_{k, \ell, j}}{w_{k, \ell}} \cdot \left( u_{k, \ell}^d - v_{k, \ell}^d \right).
\end{equation}
From this one easily computes the probability of joining a queue with length $\ell$ for which the job currently being served is in phase $j$: $J_{\ell, j} = \sum_k J_{k,\ell,j}$. Let us denote by $X$ a generic job size variable, and by $X_j$ a generic random Phase Type random variable with rate matrix $A$ which starts in phase $j$. Associated with these values we denote $X_{\ell,j}$, the convolution of $X_j$ with $\ell-1$ i.i.d.~copies of $X$.  We find that the waiting time distribution is given by:
\begin{equation}\label{eq:FWbar}
\bar F_W(w)
=
\sum_{\ell \geq 1} \sum_{j} J_{\ell, j} \bar F_{X_{\ell,j}}(w).
\end{equation}
From this it is not hard to obtain the response time distribution:
\begin{equation} \label{eq:FRbar}
\bar F_R(w)=\sum_{\ell \geq 1} \sum_{j} J_{\ell, j} \bar F_{X_{\ell+1,j}}(w) + (1-\lambda^d) \bar F_X(w).
\end{equation}

\subsection{Job Size Distribution}\label{sec:PH}
In real systems job sizes are known to be highly variable and a significant part of the total workload is often offered by a small fraction of long jobs, while the remaining workload consists mostly of short jobs (e.g., \cite{Sparrow,delgado2016job, delgado2015hawk}). A measure for the variability of a distribution is the Squared Coefficient of Variation (SCV), which is defined as $\frac{\textup{Var}(X)}{\E[X]^2}$. The SCV of an exponential random variable is exactly equal to one, while measurements in real systems reveal much higher SCVs (e.g., \cite[Chapter 20]{bookMor}). Therefore we use the class of Mixed Erlang (MErlang) distributions to investigate the performance of the proposed policies.

The parameters of the MErlang distribution are set such that we can vary the SCV in a systematic manner as well as the fraction $f$ of the workload offered by the small jobs. 
More precisely we fix a value of $k$, with probability $p$ a job is a type-$1$ job and has an Erlang($k, k\cdot \mu_1$) length with $\mu_1 > 1$ and with the remaining probability $1-p$ a job is a type-$2$ job and
has an Erlang($k, k\cdot \mu_2$) length with $\mu_2 < 1$. Hence, the type-$2$ jobs are longer on average
and we therefore sometimes refer to the type-$2$ jobs as the {\it long} jobs. 
Note that when $k=1$ the MErlang distribution is an order $2$ hyperexponential 
distribution. In order to set the parameters of the MErlang distribution, 
we first pick some value for the parameter $k$ and subsequently, the parameters $p, \mu_1$ and $\mu_2$ are set such that the following three values are matched: 
\begin{itemize}
\item the mean job length (set to one),
\item the squared coefficient of variation (SCV) and
\item the fraction $f$ of the workload that is offered by the type-$1$ jobs.
\end{itemize}
Due to \cite[Equation (4)]{fang2001hyper} one finds that for 
$p, \mu_1$ and $\mu_2$ fixed, the SCV as a function of $k$ is such that
$(SCV(k)+1)/(1+1/k)$ is a constant equal to $(p/\mu_1^2+(1-p)/\mu_2^2)/(p/\mu_1+(1-p)/\mu_2)^2$.
Hence, if we define $\widetilde{SCV} = 2\frac{SCV+1}{1+\frac{1}{k}}-1$,
we can match $(1,SCV,f)$ for a general $k$ by matching
$(1,\widetilde{SCV},f)$ for a hyperexponential distribution 
as is \cite[Section 7.1]{hellemans2018power},
that is , set $\mu_1, \mu_2$ as:
\begin{align*}
 \mu_{i} &= \frac{\widetilde{SCV}+(4f-1)+(-1)^{i-1} \sqrt{(\widetilde{SCV}-1)(\widetilde{SCV}-1+8f\bar f)}}{2f(\widetilde{SCV}+1)},
\end{align*}
with $\bar f = 1-f$  and $p = \mu_1 f$.

 We note that this distribution has a PH representation with $\alpha \in \R^{2\cdot k}$ defined by $\alpha_1 = p$, $\alpha_{k+1}=1-p$ and $\alpha_i=0$ for $i \neq 1, k+1$. Furthermore we define the transition matrix $A \in \R^{2\cdot k}$ by stating its non-zero elements: $A_{i,i}=-k \cdot \mu_1, A_{k+i, k+i} = -k\mu_2$ for $i=1,\dots,k$ and $A_{i,i+1} = k\cdot \mu_1$, $A_{k+i, k+i+1} = k \cdot \mu_2$ for $i=1,\dots,k-1$. Throughout, we denote this job size distribution as MErlang($SCV, f, k$) and as HEXP($SCV, f$) when $k=1$.

\begin{table*}
	\caption{Comparison of mean response time for the finite system and the cavity method.} \label{table:finite_accuracy}
	\[
	\begin{array}{@{}l*{7}{c}@{}}
		\toprule
		\text{Policy} & \multicolumn{7}{c@{}}{\text{N}}\\
		\cmidrule(l){2-8}
		& 10 & 50 & 100 & 500 & 1000 & 2000 & \infty\\
		\midrule
		\text{SQ(3)-RTB} & 1.6204 & 1.0409 & 0.9829 & 0.9278 & 0.9217 & 0.9182  & 0.9172\\
		\text{SQ(5)-RE(2)} & 3.3083 &  1.8969  & 1.7229 &  1.6017 &  1.5810 &  1.5656 &  1.5649 \\
		\text{SQ(10)-RTB-RE(2)} & 2.1140 &  0.3635  & 0.2374 &  0.1565 &  0.1452 &  0.1384 &  0.1366  \\
		LAS(2) & 4.8688 &  3.9673  & 3.850 &  3.7789 &  3.7743 &  3.7551 &  3.7156 \\
		\text{LAS(7)-QTB} & 7.6238 &  1.4891  & 0.9909 &  0.6986 &  0.6824 &  0.6523 &  0.6462  \\
		\text{RE(6, 2)} & 3.4524 &  1.7820  & 1.5838 &  1.4292 &  1.4087 &  1.3934 &  1.3846 \\
		\text{LEW(8)} & 6.0683 &  1.5494  & 1.1312 &  0.8801 &  0.8484 &  0.8275 &  0.8266\\
		\bottomrule
	\end{array}
	\]
\end{table*}

\section{Finite System Accuracy} \label{sec:finiteAccuracy}

In this section we compare the mean waiting time that corresponds to the
fixed point of the cavity map with simulation experiments. The simulation setup
is identical to the model, except that the number of servers $N$ is finite.  
For each policy in Section \ref{sec:examples} we selected some arbitrary
parameter setting and varied the number of servers from $N=10$ to
$N=2000$. All simulation runs simulate the system up to time $t=10^7/N$ and use a warm-up period of $30\%$. Each simulation is the mean of $40$ runs. The simulation experiments were performed using the following parameter settings:
\begin{itemize}
\item SQ($3$)-RTB : $\lambda = 0.7$, $\Delta=0.01$ and MErlang($10, 1/2, 2$) job sizes.
\item SQ($5$)-RE($2$) : $\lambda = 0.8$ and HEXP($10, 1/10$) job sizes.
\item SQ($10$)-RTB-RE($2$) : $\lambda = 0.8$, $\Delta=0.1$ and MErlang($15, 1/3, 5$) job sizes.
\item LAS($2$) : $\lambda = 0.6$, $\Delta=0.5$ and HEXP($20, 1/4$) job sizes.
\item LAS($7$)-QTB : $\lambda = 0.9$, $\Delta=0.01$ and MErlang($20, 2/3, 5$) job sizes.
\item RE($6, 2$) : $\lambda = 0.8$ and HEXP($10,1/4$) job sizes.
\item LEW($8$) : $\lambda = 0.9$, $\Delta=0.5$ and HEXP($15, 1/3$) job sizes.
\end{itemize}
The results are summarized in Table \ref{table:finite_accuracy}. We observe that for all examples there is indeed convergence of the mean waiting time, rendering our models to be accurate for high values of $N$.

\section{Numerical Experiments} \label{sec:numExperiments}
Throughout this section we compare the proposed policies relative to the SQ($d$) policy. Therefore, we are interested in the relative improvement of the proposed policies, in particular we focus on the quantity:
\begin{equation}\label{eq:EWrel}
E_{W, rel, P} = \frac{\E[W^{(SQ(d))}] - \E[W^{(P)}]}{\E[W^{(SQ(d))}]},
\end{equation}
where $\E[W^{(P)}]$ denotes the expected waiting time for some policy $P$. This value denotes how much additional waiting time using the SQ($d$) policy yields. Note that if one was mainly interested in the relative improvement in the response time, one would compute:
$$
E_{R, rel, P} = \frac{\E[R^{(SQ(d))}] - \E[R^{(P)}]}{\E[R^{(SQ(d))}]} = \frac{\E[W^{(SQ(d))}] - \E[W^{(P)}]}{\E[W^{(SQ(d))}] + 1},
$$
which is simply a flattened version of $E_{W, rel, P}$. As $\E[R]=\E[W]+1$, as we are mainly interested in the amount \textit{delay} a job experiences, we focus on $E_{W, rel, P}$ as defined in \eqref{eq:EWrel}.
\begin{figure*}[t]
\begin{subfigure}{0.4\textwidth}
\centering
\captionsetup{width=.8\linewidth}
\includegraphics[width=1\linewidth]{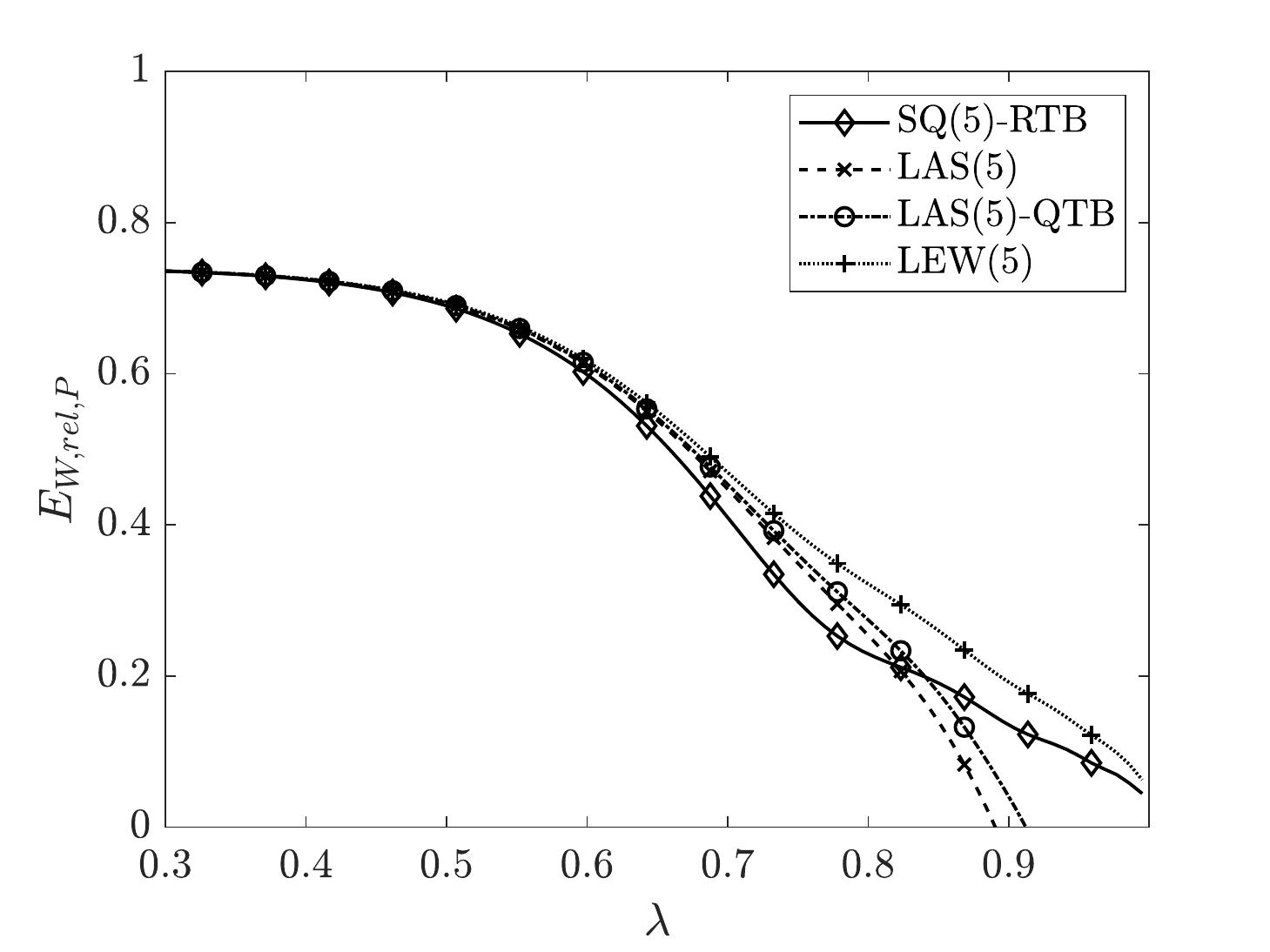}
\caption{Policies which do not make use of any threshold value.}
\label{fig2a}
\end{subfigure}
\begin{subfigure}{.4\textwidth}
\centering
\captionsetup{width=.8\linewidth}
\includegraphics[width=1\linewidth]{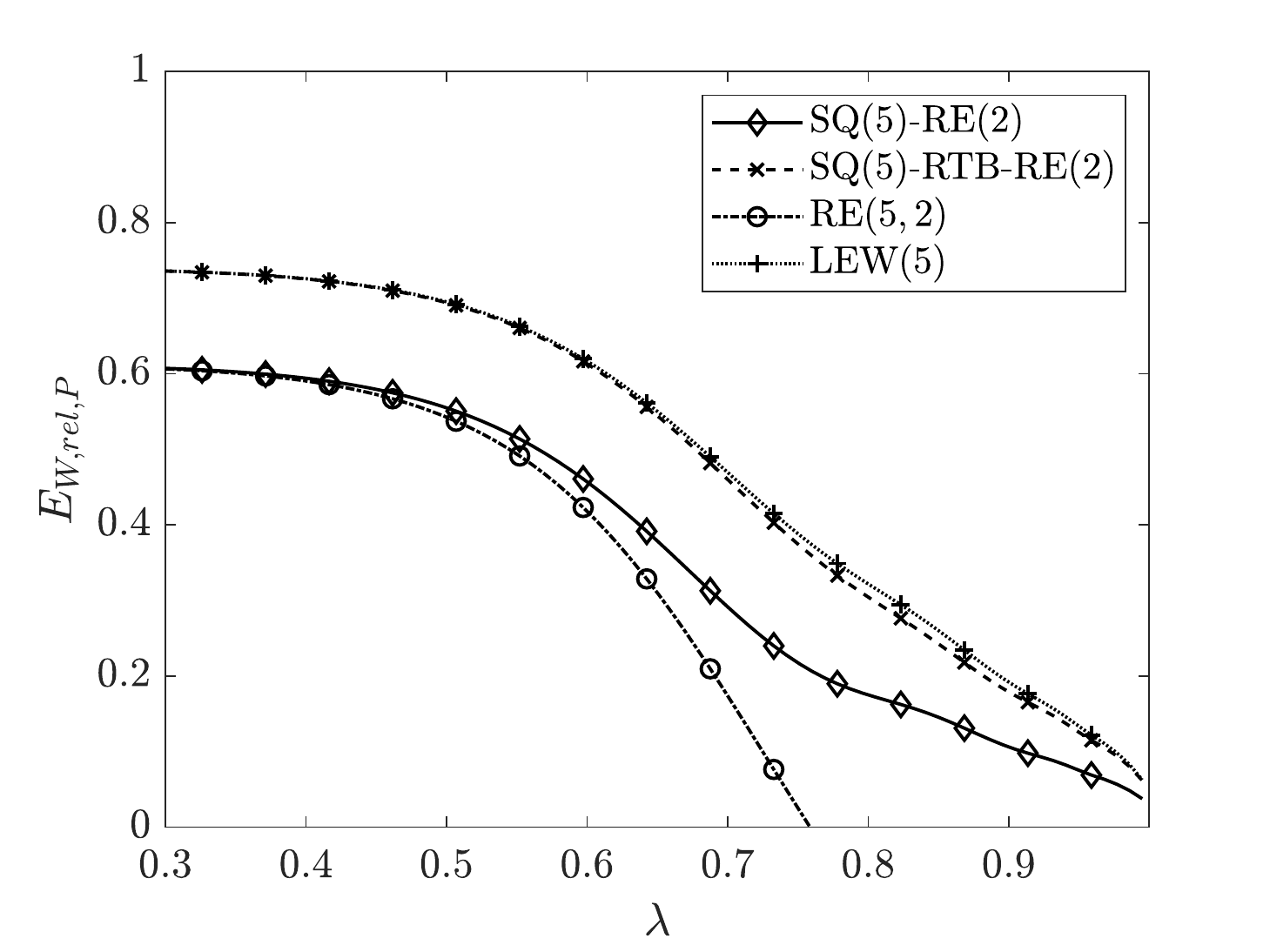}
\caption{Policies which make use of the threshold value $T=2$.}
\label{fig2b}
\end{subfigure}
\caption{Plots of the improvement in mean waiting time (see also \eqref{eq:EWrel}) as a function of $\lambda$ for $d=5$, $\Delta=0.1$ and HEXP($10,1/2$) job sizes.}
\label{fig2}
\end{figure*}
Throughout our numerical experimentation, we employ one central example as a base case in order to investigate the effect different parameters have on $E_{W, rel, P}$ for the policies $P$ discussed in Section \ref{sec:examples}. As the base case, we take $d=5$, $\Delta=0.1$ and HEXP($10,1/2$) distributed job sizes. In Figure \ref{fig2} we observe the evolution of $E_{W, rel, P}$ as a function of the arrival rate $\lambda$. From  Figure \ref{fig2}, we can already make quite a few observations:
\begin{itemize}
\item The relative improvement we obtain from using the attained service time information is significant, with relative improvements in the waiting time close to $75\%$. However, the improvement decreases as the arrival rate $\lambda$ approaches one. This makes sense as for a higher value of $\lambda$ the queues become longer and queue length information becomes more important than attained service time information.
\item All policies which mainly focus on the queue length information improve the performance
for all $\lambda$, as such they can be viewed as enhanced SQ($d$) policies. For the policies which mainly use the attained service time to distribute jobs, we observe that there exists some $\lambda_{\max}$ such that these policies outperform SQ($d$) for all $\lambda < \lambda_{\max}$ while they are outperformed by SQ($d$) when $\lambda > \lambda_{\max}$.  
\item Many of our policies perform as good as LEW($d$) for loads up to $0.6$, while these 
policies do not make use of the job size distribution.
The performance of the SQ($5$)-RTB-RE($2$) policy is very close to that of LEW($d$) 
for all $\lambda$, where the threshold $T=2$ was chosen quite arbitrarily (cf. Section \ref{sec:num_threshold}).
\item The improvement has a plateau at first, then the relative improvement tends to decrease with $\lambda$. For the policies which mainly employ the queue length information to distribute jobs, the curves
become somewhat irregular for high loads. This irregularity is discussed when we
consider the impact of larger $d$ values (see Section \ref{sec:num_d}).
\item The low traffic limit ($\lambda\rightarrow 0^+$) appears to be the same for all policies which solely rely on a threshold and those which rely on the attained service time information. This makes sense as in the low traffic limits, whenever a job has a non-zero waiting time, all $d$ chosen queues  have exactly one job waiting, meaning queue length information becomes irrelevant.
\end{itemize}

\begin{remark}
As $\lambda$ increases, the number of jobs in each queue increases. This reduces the value of knowing the attained service time of the job at the head of the queue. However, if we were to use a scheduling policy such as Processor Sharing, dispatchers may request information on the attained service time of all jobs in the queue, which could result in a more significant improvement at high loads.
\end{remark}
In the following sections, we reproduce Figure \ref{fig2}, where we change the value of one parameter. This allows us to investigate the effect this parameter has on our basic example.

\begin{figure*}[t]
\begin{subfigure}{0.4\textwidth}
\centering
\captionsetup{width=.8\linewidth}
\includegraphics[width=1\linewidth]{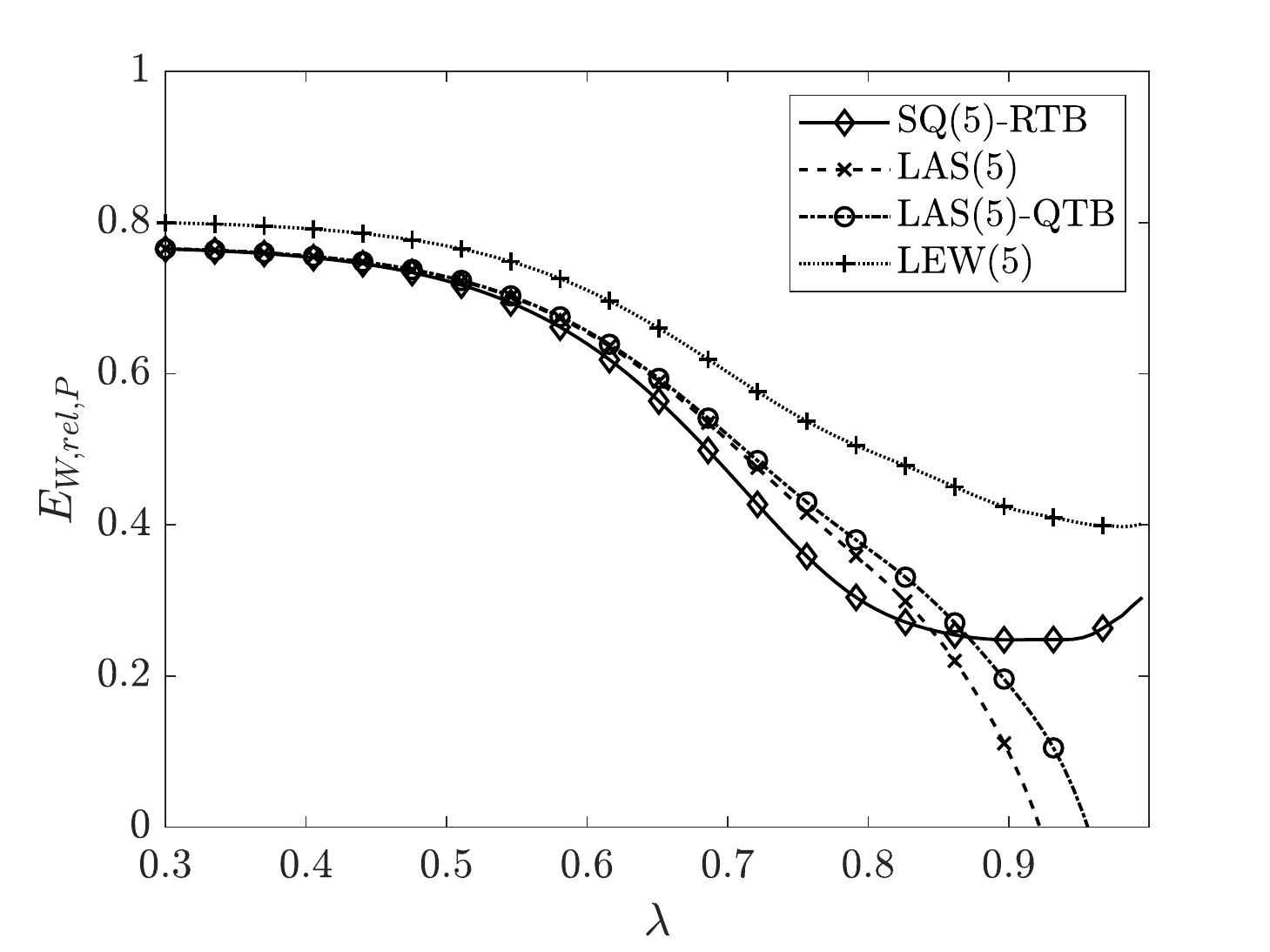}
\caption{Policies which do not make use of any threshold value, this figure should be compared to Figure \ref{fig2a}.}
\label{fig14a}
\end{subfigure}
\begin{subfigure}{.4\textwidth}
\centering
\captionsetup{width=.8\linewidth}
\includegraphics[width=1\linewidth]{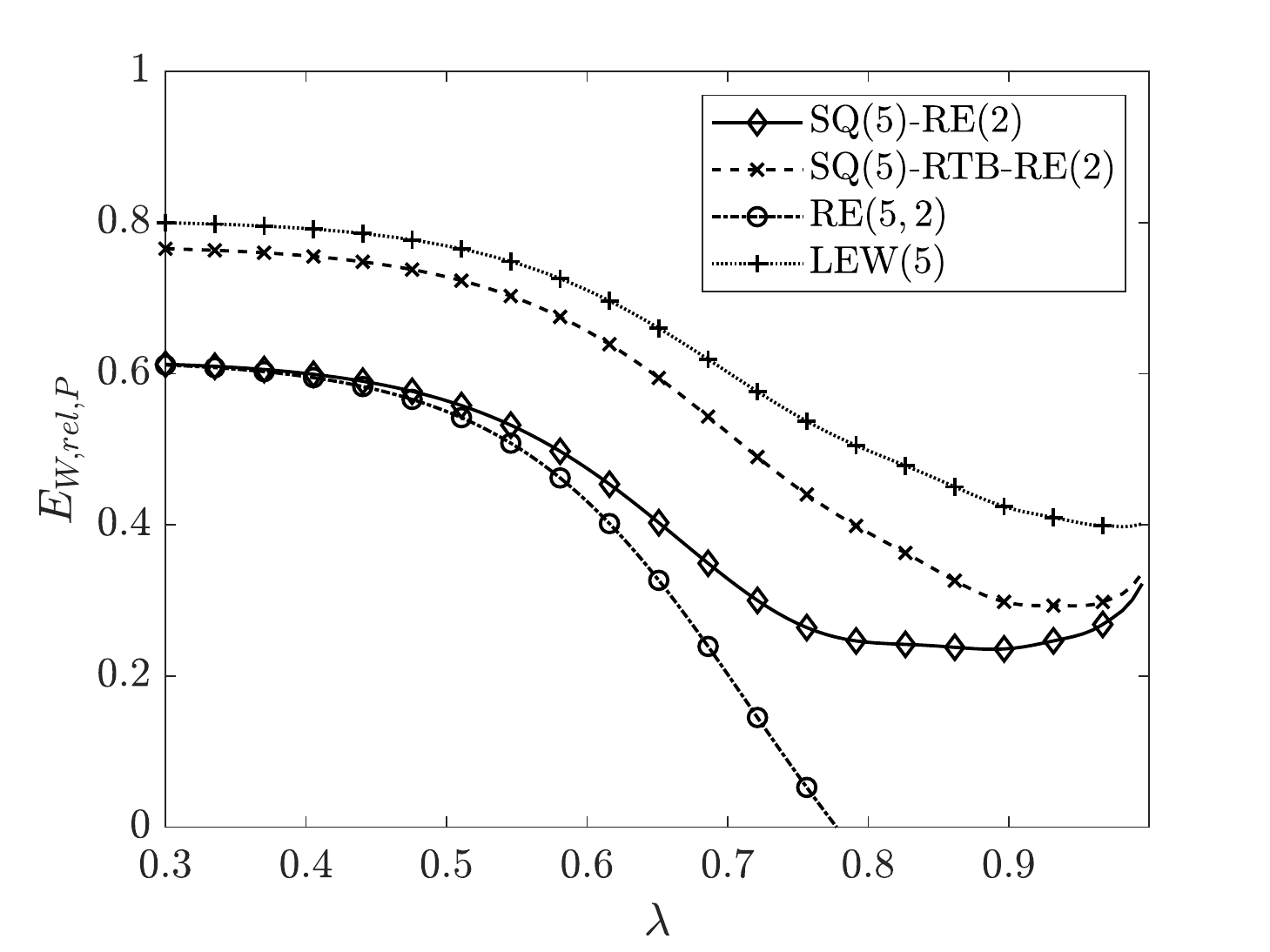}
\caption{Policies which make use of the threshold value $T=2$, this figure should be compared to Figure \ref{fig2b}.}
\label{fig14b}
\end{subfigure}
\begin{subfigure}{.4\textwidth}
\centering
\captionsetup{width=.8\linewidth}
\includegraphics[width=1\linewidth]{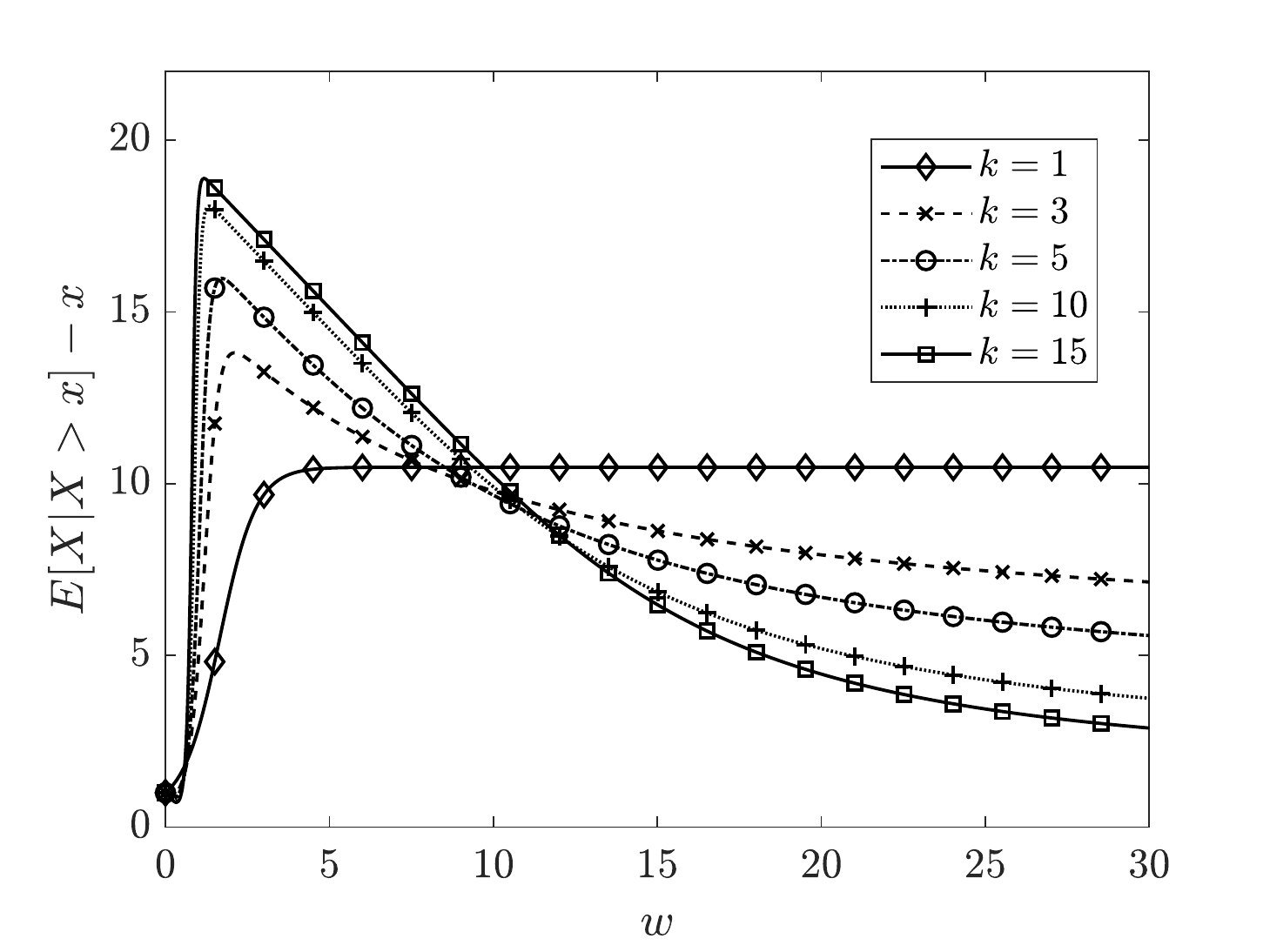}
\caption{Expected remaining service time for a MErlang($10,1/2,k$) distribution with $k=1, 3, 5, 10$ and $15$.}
\label{fig14c}
\end{subfigure}
\caption{Plots  of the improvement in mean waiting time as a function of $\lambda$ for $d=5$, $\Delta=0.1$, MErlang($10,1/2,\textbf{5}$) job sizes. Along with a plot of the conditional expected remaining service time for the MErlang($10,1/2,\textbf{k}$) distribution.}
\label{fig14}
\end{figure*}

\subsection{Impact of the number of phases $k$}
In Figures \ref{fig14a} and \ref{fig14b} we repeat the experiment from Figure \ref{fig2} with MErlang($10, 1/2, \textbf{5}$) job sizes. We clearly observe a greater improvement in performance than for the case $k=1$. This may seem counter-intuitive as one could argue that our policies are tailor-made for job sizes which have a decreasing hazard rate. Indeed, when the hazard rate is decreasing, the expected remaining workload increases as a function of the attained service time. This is in correspondence with $\xi(k,l)$ being non-decreasing in $k$ for our policies. In Figure \ref{fig14c} we observe that the MErlang($SCV, f, k$) only has a decreasing hazard rate for $k = 1$. However, as $k$ increases both the small and large jobs become 
less variable and thus have a more predictable size. This implies that large jobs
are somewhat easier to detect based on the attained service time, which explains why we have a greater improvement than with $k=1$. Also, the gap between our policies and the LEW($d$) policy increases. This is natural as it is sometimes better to assign a job to a server with a lower attained service time (as the hazard rate is non-decreasing), which is something the LEW($d$) does automatically while this is never done by the other policies considered.

\subsection{Impact of the number of chosen servers $d$} \label{sec:num_d}
\begin{figure*}[t]
\begin{subfigure}{0.4\textwidth}
\centering
\captionsetup{width=.8\linewidth}
\includegraphics[width=1\linewidth]{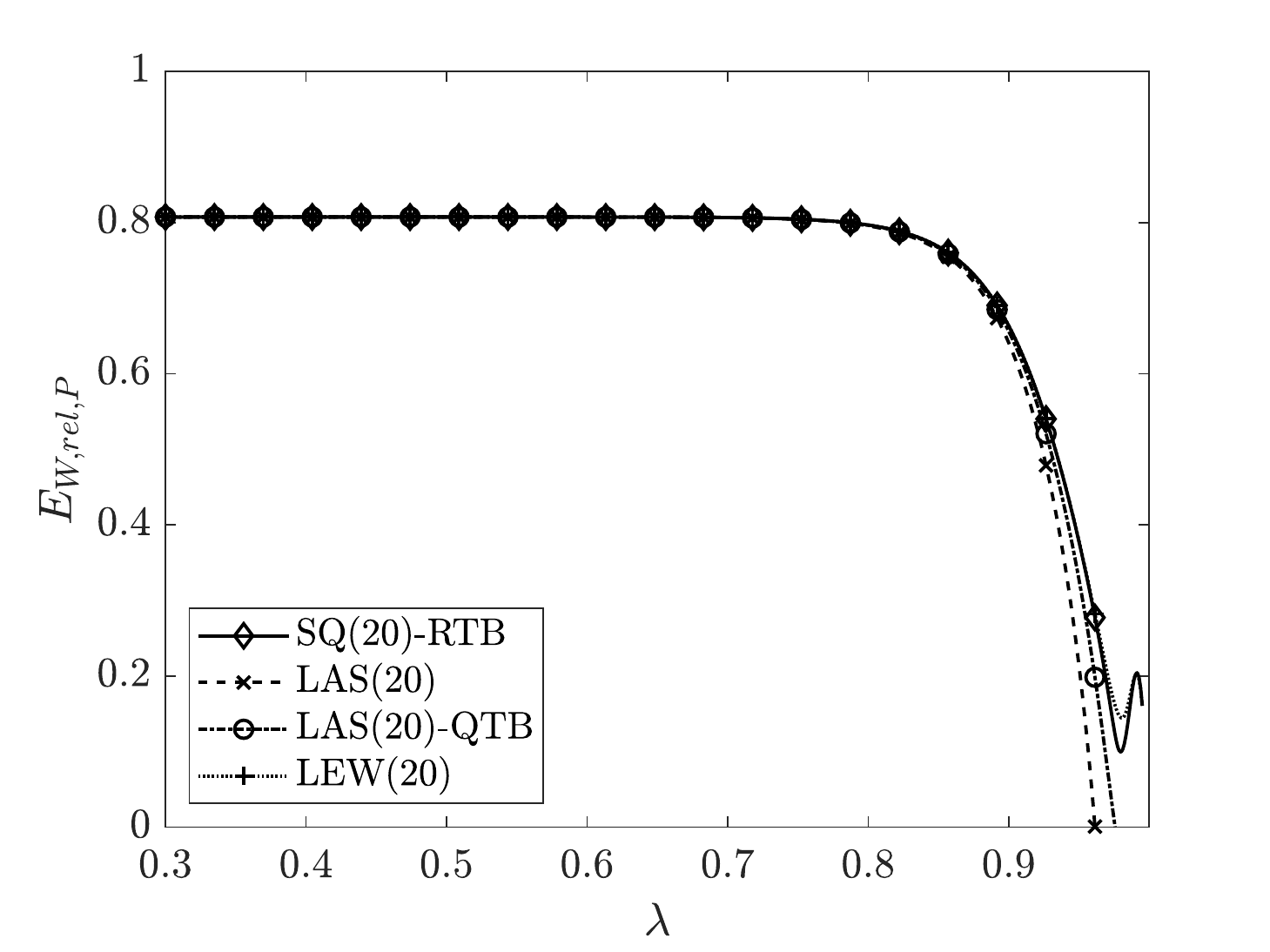}
\caption{Policies which do not make use of any threshold value.}
\label{fig5a}
\end{subfigure}
\begin{subfigure}{.4\textwidth}
\centering
\captionsetup{width=.8\linewidth}
\includegraphics[width=1\linewidth]{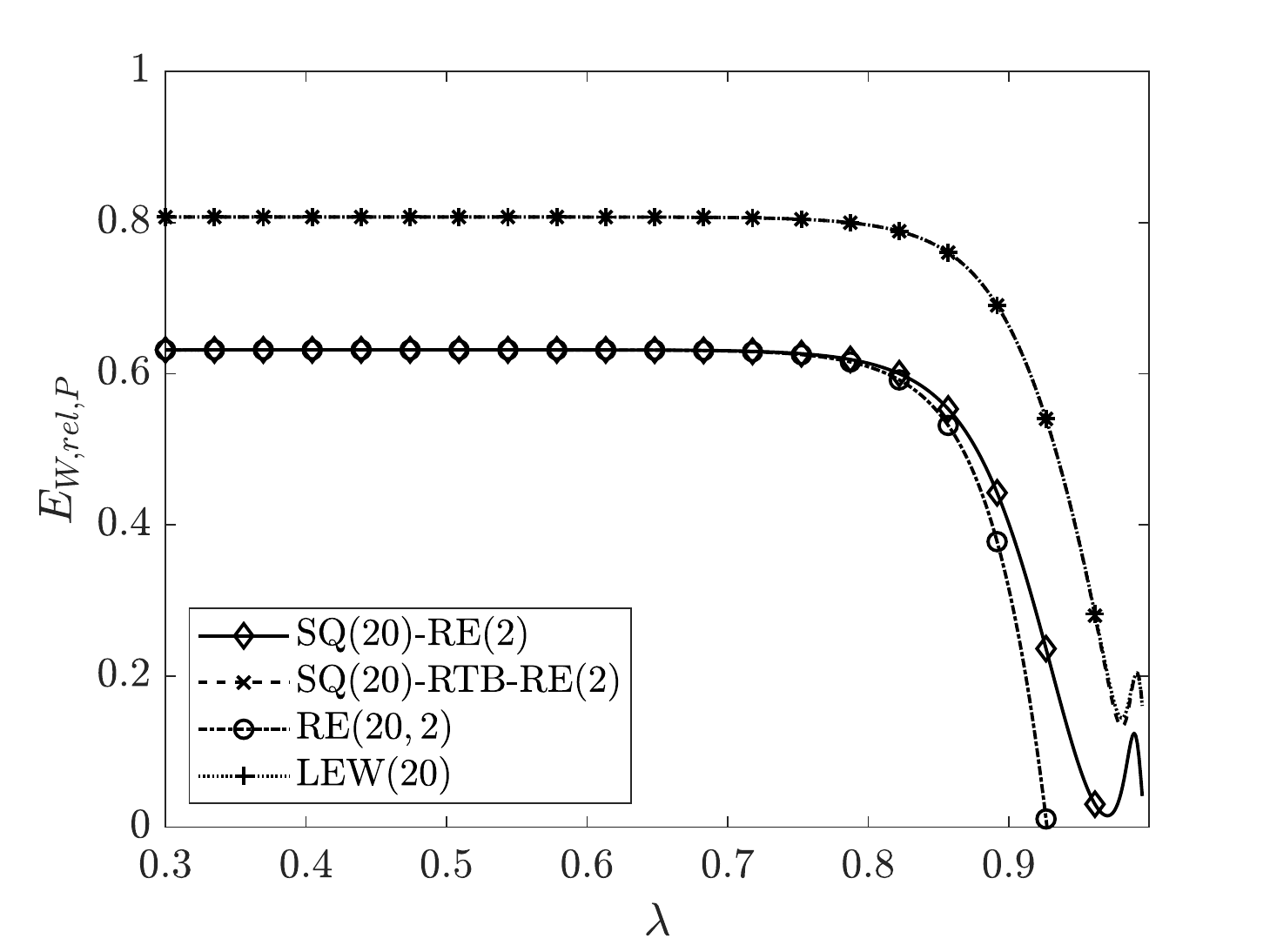}
\caption{Policies which make use of the threshold value $T=2$.}
\label{fig5b}
\end{subfigure}
\caption{Plots of the improvement in mean waiting time as a function of $\lambda$ for $\textbf{d=20}$, $\Delta=0.1$ and HEXP($10,1/2$) job sizes. This figure should be compared to Figure \ref{fig2}.}
\label{fig5}
\end{figure*}
In Figure \ref{fig5} we reproduce Figure \ref{fig2} after setting $d=20$ rather than $d=5$ (and thus setting $\Delta=0.1$, $SCV=10$ and $f=1/2$). We observe that the relative improvement increases significantly by increasing the value of $d$. In particular we make the following observations:
\begin{itemize}
\item Increasing the value of $d$ increases the low traffic limit and extends the range of
the load at which this value is maintained. This is expected as higher $d$ values
result in shorter queues.
\item For the same reason, the value of $\lambda_{max}$ 
up to which the policies which solely depend on the 
attained service time outperform SQ($d$) increases for larger $d$ values.
\item The gap between the SQ($d$)-RTB-RE($2$) and LEW($d$) becomes negligible by increasing the value of $d$ to $20$.
\item For some policies, the behaviour becomes irregular as $\lambda$ approaches one. To understand the cause of this irregularity, we define $\mathcal{T}_d(\lambda)$ as the expected number of jobs which have the least number of jobs pending amongst $d$ randomly selected servers (given that all $d$ chosen servers are busy). Letting $u_k$ denote the probability that a server has $k$ or more pending jobs, we find that:
\begin{equation} \label{eq:Tdlam}
\mathcal{T}_d(\lambda)=\frac{\sum_{k=1}^d k \cdot \sum_{\ell=1}^\infty p_{\ell,k} }{u_1^d},
\end{equation}
with $p_{\ell,k}=\binom{d}{k} (u_\ell - u_{\ell+1})^k \cdot u_{\ell+1}^{d-k}$ the probability that $k$ servers have exactly $\ell$ pending jobs and all other selected servers have at least $\ell+1$ jobs in their queue. In Figure \ref{fig6} we plot the evolution of $T_d(\lambda)$ for the SQ($d$) policy with various values of $d$, $SCV=10$ and $f=1/2$ in Figure \ref{fig6a} and $f=1/10$ in Figure \ref{fig6b}. We observe the same type of irregularity as in Figure \ref{fig5}. For $\mathcal{T}_d(\lambda)$ this behaviour can be understood as follows:
for low loads the $d$ chosen servers have queue length $1$ and the number of chosen servers with queue length one then decreases as $\lambda$ increases. 
If we look at the mean number of selected servers with queue length $2$,
then this mean increases with $\lambda$ (except for very high $\lambda$). For larger 
$\lambda$ values, the expected number of chosen servers with the least number of jobs 
also depends on the mean number of servers with queue length $2$ as all $d$ 
chosen servers
may have a queue length of at least $2$. This causes the \textquotedblleft waves\textquotedblright\ close to $\lambda=1$ which become more pronounced as $d$ grows. These waves also explain the irregularity in Figure \ref{fig5} (and Figure \ref{fig2}) as a higher value for $T_d(\lambda)$ implies we have more ties in the queue length and the attained service time information
is more valuable. In Figure \ref{fig6b} we observe that decreasing the value of $f$ decreases the 
height of the waves.
\end{itemize}
To further emphasize the aforementioned observations, we single out two policies which were used to create Figure \ref{fig5} and show their performance as a function of $\lambda$ for various values of $d: 2,5,10,15,20$. In Figure \ref{fig7a} we consider the SQ($d$)-RE($2$) policy whilst in Figure \ref{fig7b} we consider the RE($d,2$) policy. We can clearly see our observations being confirmed. In particular for SQ($d$)-RE($2$), we see waves becoming larger as $d$ increases, the plots of $d=15$ and $d=20$ even cross at some point.
\begin{figure*}[t]
\begin{subfigure}{0.4\textwidth}
\centering
\captionsetup{width=.8\linewidth}
\includegraphics[width=1\linewidth]{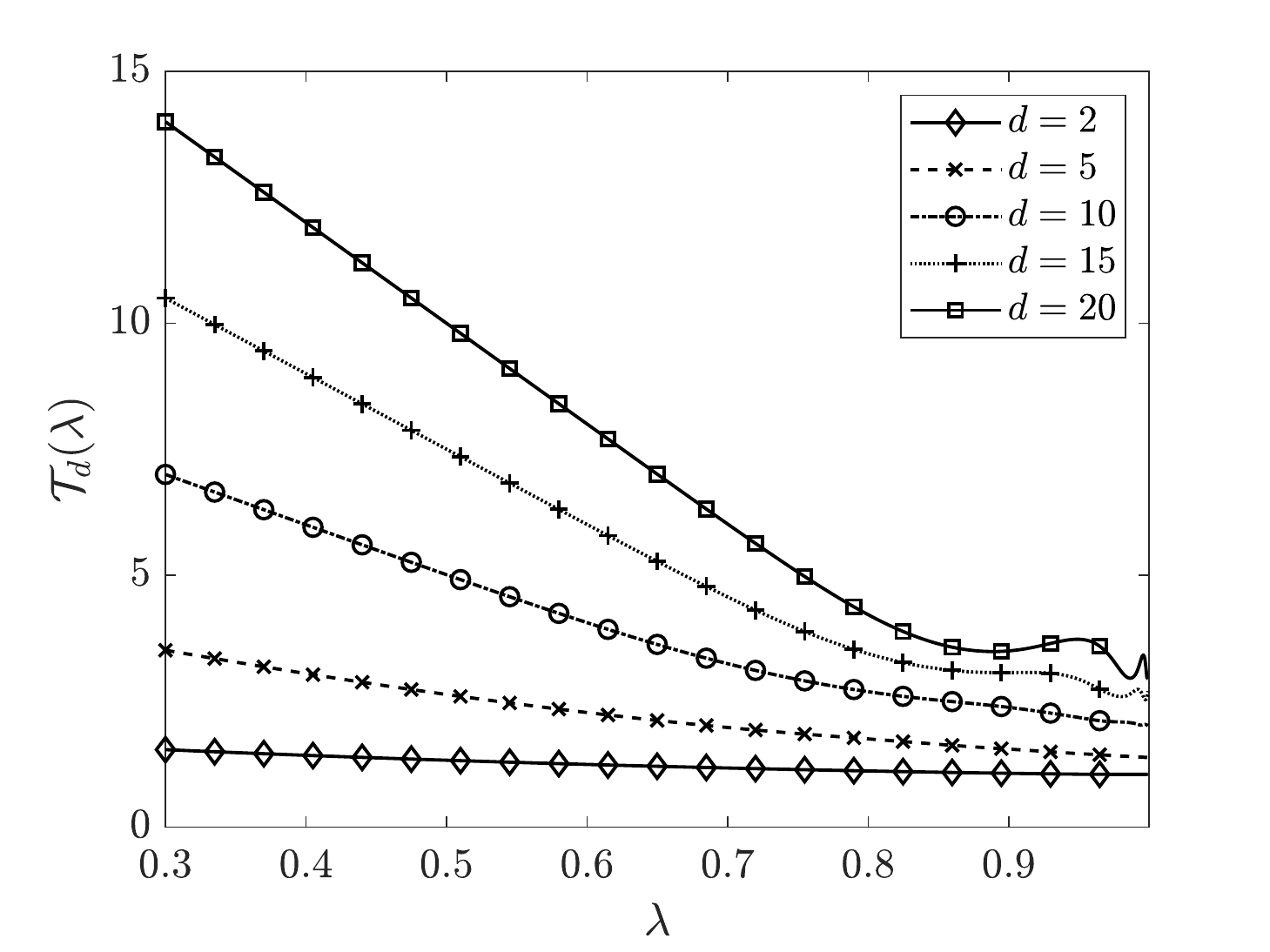}
\caption{Here we set $f=1/2$, this figure is used to motivate the behaviour in Figure \ref{fig5} and \ref{fig7} for $\lambda \approx 1$.}
\label{fig6a}
\end{subfigure}
\begin{subfigure}{.4\textwidth}
\centering
\captionsetup{width=.8\linewidth}
\includegraphics[width=1\linewidth]{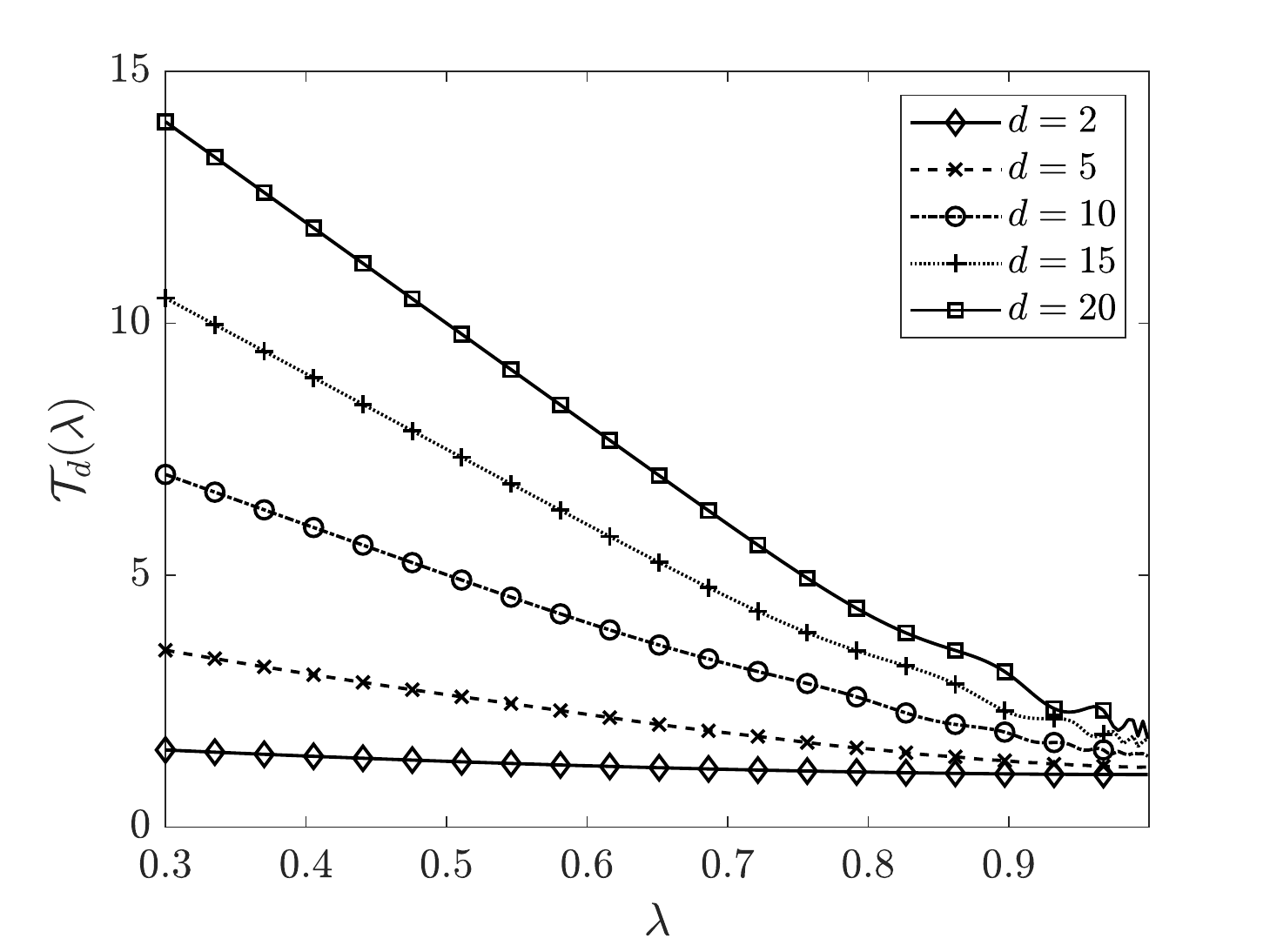}
\caption{Here we set $f=1/10$, as such this setting can be seen to correspond to that considered in Figure \ref{fig8}.}
\label{fig6b}
\end{subfigure}
\caption{Plots of $\mathcal{T}_d(\lambda)$ as a function of $\lambda$ for the SQ($d$) policy with various values of $d$ and HEXP($10, f$) job sizes.}
\label{fig6}
\end{figure*}

\begin{figure*}[t]
\begin{subfigure}{0.4\textwidth}
\centering
\captionsetup{width=.8\linewidth}
\includegraphics[width=1\linewidth]{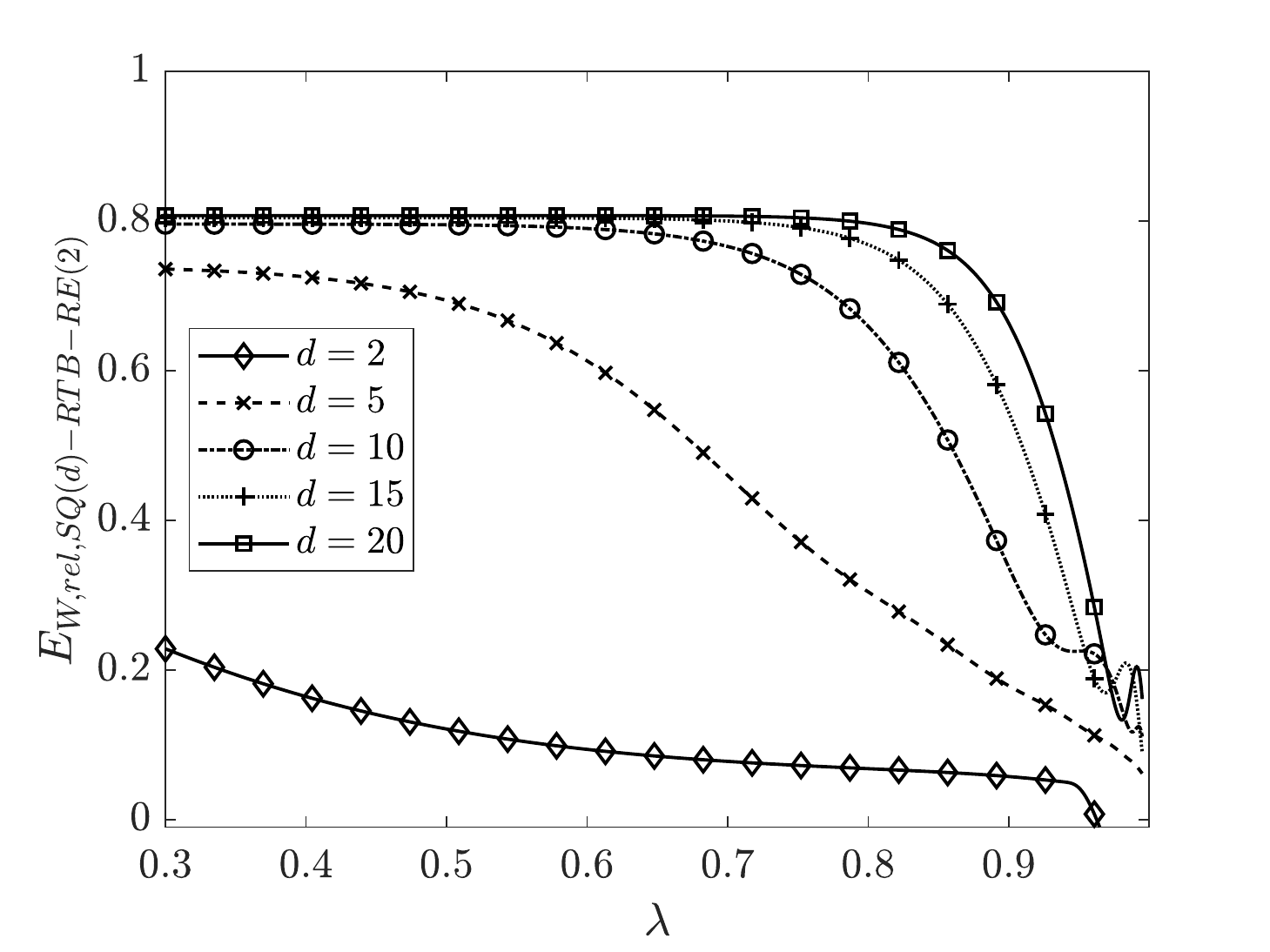}
\caption{The SQ($d$)-RTB-RE($2$) policy.}
\label{fig7a}
\end{subfigure}
\begin{subfigure}{.4\textwidth}
\centering
\captionsetup{width=.8\linewidth}
\includegraphics[width=1\linewidth]{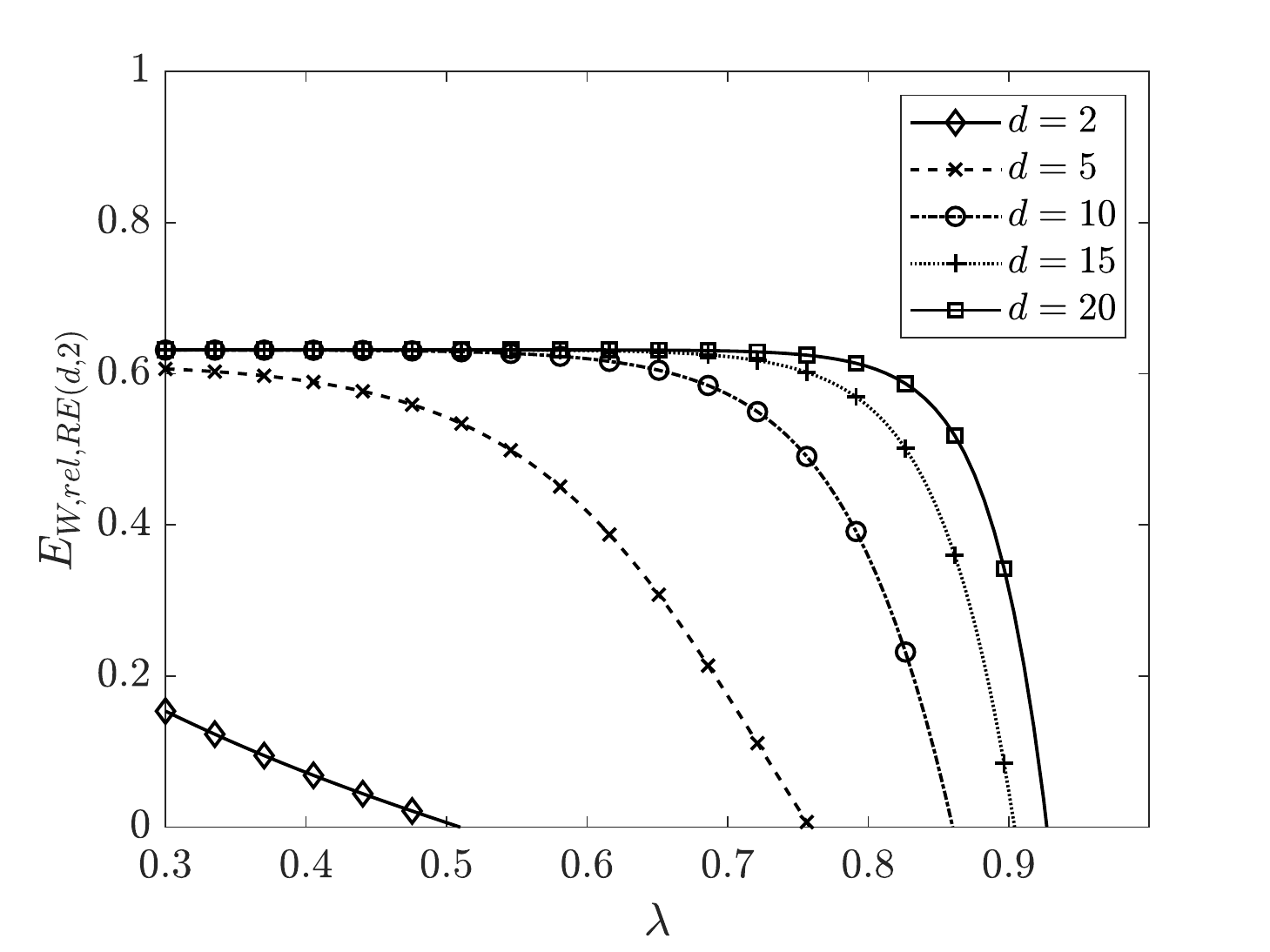}
\caption{The RE($d, 2$) policy.}
\label{fig7b}
\end{subfigure}
\caption{Plots of the improvement in mean waiting time as a function of $\lambda$ for various values of $d$, $\Delta=0.1$ and HEXP($10, 1/2$) job sizes. This figure should be viewed as a supplement of Figure \ref{fig5}.}
\label{fig7}
\end{figure*}

\subsection{Impact of the Squared Coefficient of Variation ($SCV$)}
\begin{figure*}[t]
\begin{subfigure}{0.4\textwidth}
\centering
\captionsetup{width=.8\linewidth}
\includegraphics[width=1\linewidth]{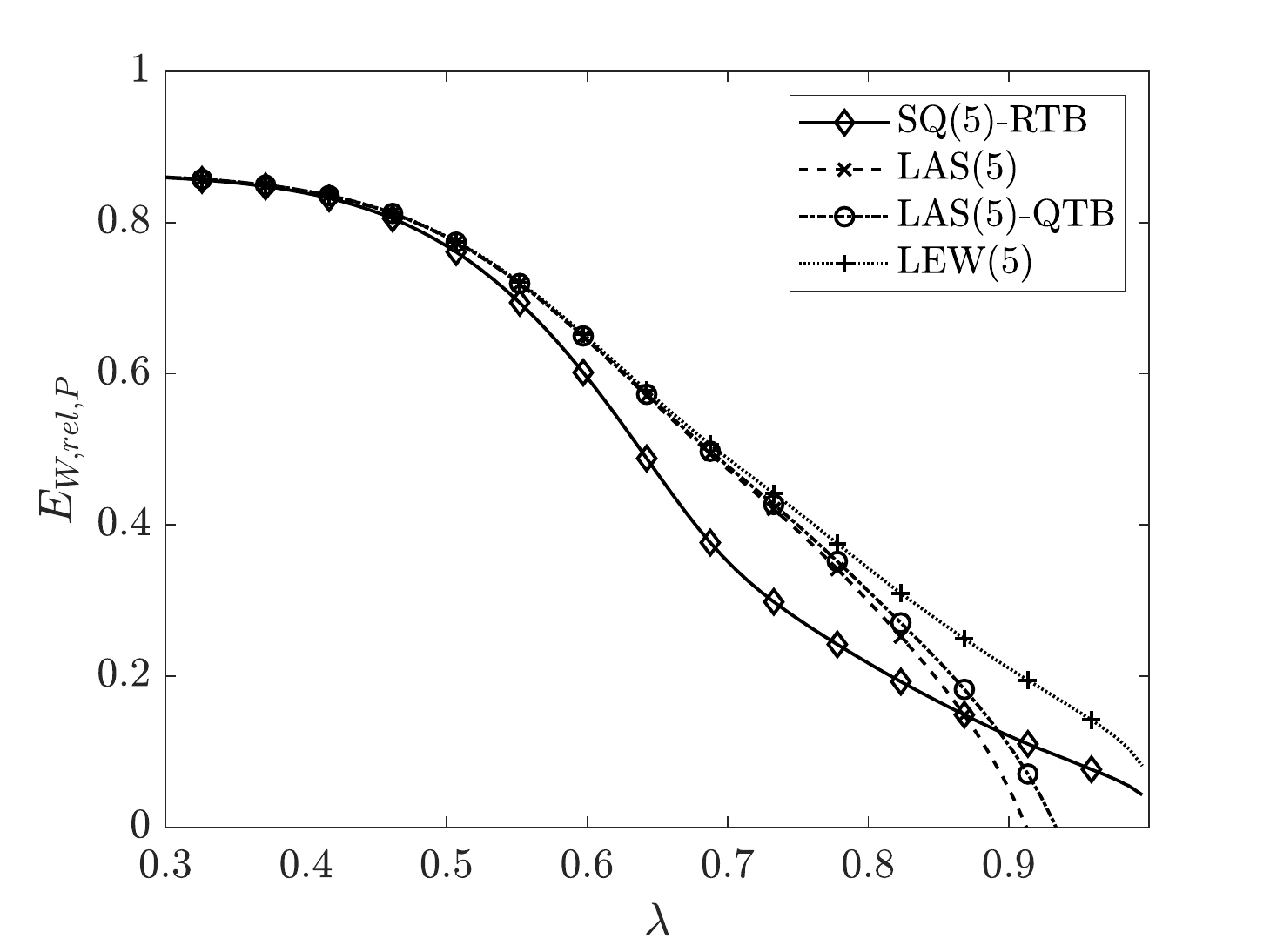}
\caption{Policies which do not make use of any threshold value.}
\label{fig3a}
\end{subfigure}
\begin{subfigure}{.4\textwidth}
\centering
\captionsetup{width=.8\linewidth}
\includegraphics[width=1\linewidth]{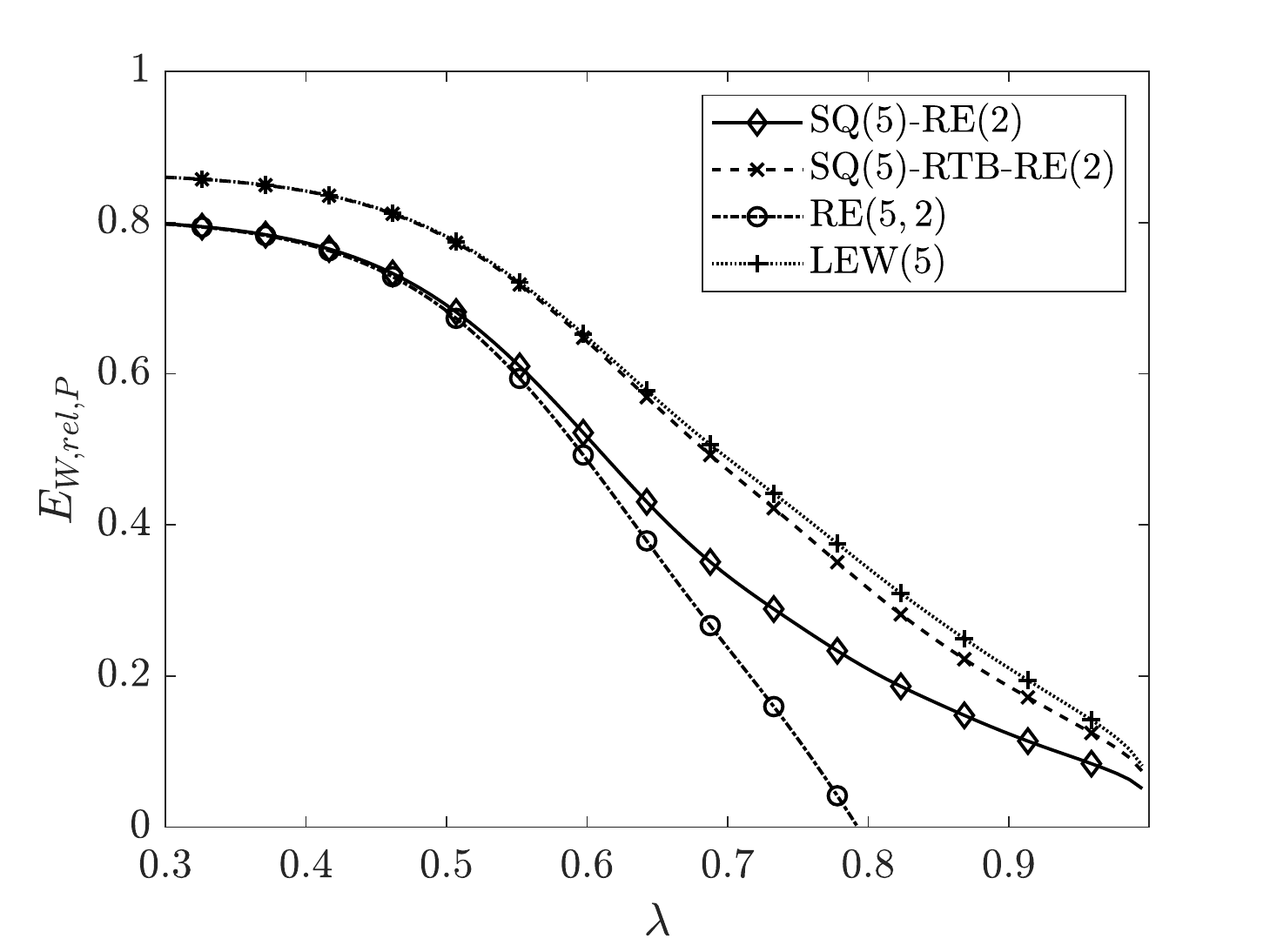}
\caption{Policies which make use of the threshold value $T=2$.}
\label{fig3b}
\end{subfigure}
\caption{Plots of the improvement in mean waiting time as a function of $\lambda$ for $d=5$, $\Delta=0.1$ and HEXP($\textbf{30},1/2$) job sizes. This figure should be compared to Figure \ref{fig2}.}
\label{fig3}
\end{figure*}
\begin{figure*}[t]
\begin{subfigure}{0.4\textwidth}
\centering
\captionsetup{width=.8\linewidth}
\includegraphics[width=1\linewidth]{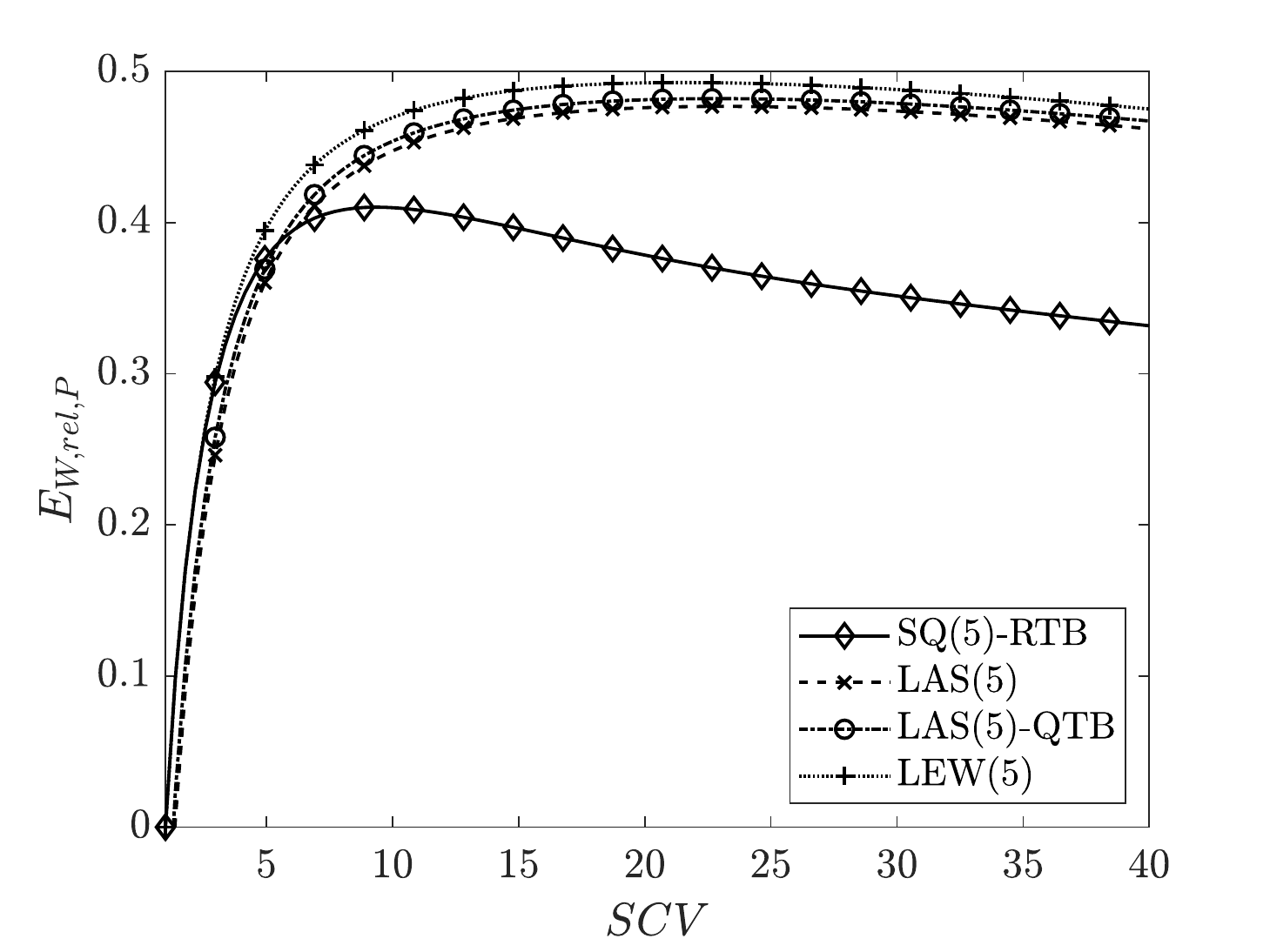}
\caption{Policies which do not make use of any threshold value.}
\label{fig4a}
\end{subfigure}
\begin{subfigure}{.4\textwidth}
\centering
\captionsetup{width=.8\linewidth}
\includegraphics[width=1\linewidth]{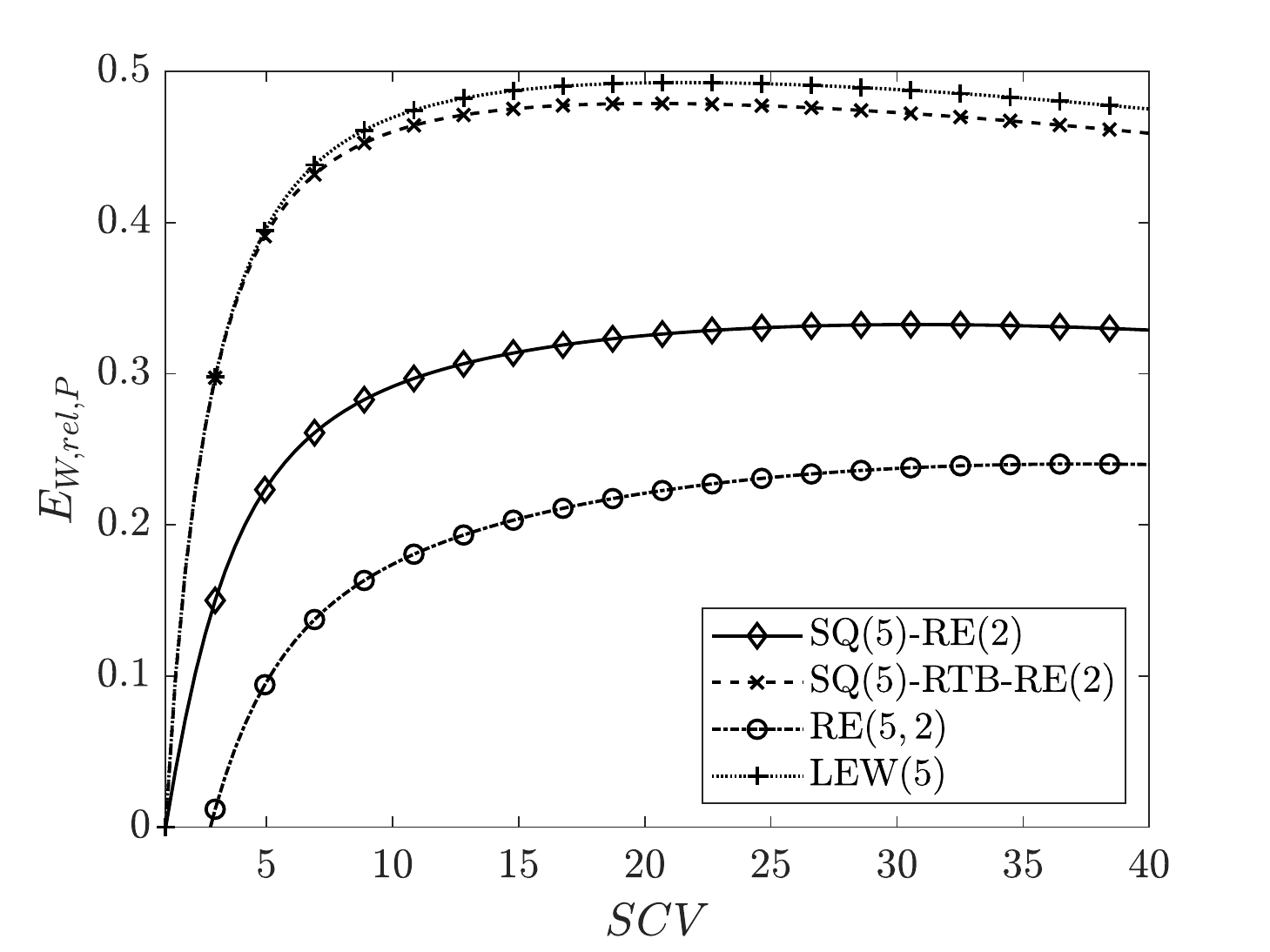}
\caption{Policies which make use of the threshold value $T=2$.}
\label{fig4b}
\end{subfigure}
\caption{Plots of the improvement in mean waiting time as a function of the $SCV$ for 
$\lambda=0.7$, $d=5$, $\Delta=0.1$ and HEXP($\textbf{SCV},1/2$) job sizes.}
\label{fig4}
\end{figure*}
In Figure \ref{fig3} we use the same parameter settings as in Figure \ref{fig2}, but change the $SCV$ from $10$ to $30$ (i.e.~we set $d=5$, $\Delta=0.1$ and $f=1/2$). As expected, increasing the SCV, increases the relative improvement made by using the attained service time information. However, the improvement is not as significant as one might expect, especially for larger values of $\lambda$. Therefore we investigate this further in Figure \ref{fig4}, where we show the relative improvement as a function of the SCV for $\lambda=0.7$, we observe that there is a clear improvement if we have a higher SCV, but only up to some point. For most policies we observe a strong improvement until $SCV \approx 15$, after which the improvement seems to flatten, or even decrease. The reason for the slight decrease is probably due to the fact that a higher SCV
also results in longer queues and similar to higher $\lambda$ values, this decreases the value
of knowing the attained service time somewhat. Additional experiments showed that for smaller $\lambda$ values the decrease as a function of the SCV occurs further on.

\subsection{Impact of the Load from Small Jobs $f$}
\begin{figure*}[t]
\begin{subfigure}{0.4\textwidth}
\centering
\captionsetup{width=.8\linewidth}
\includegraphics[width=1\linewidth]{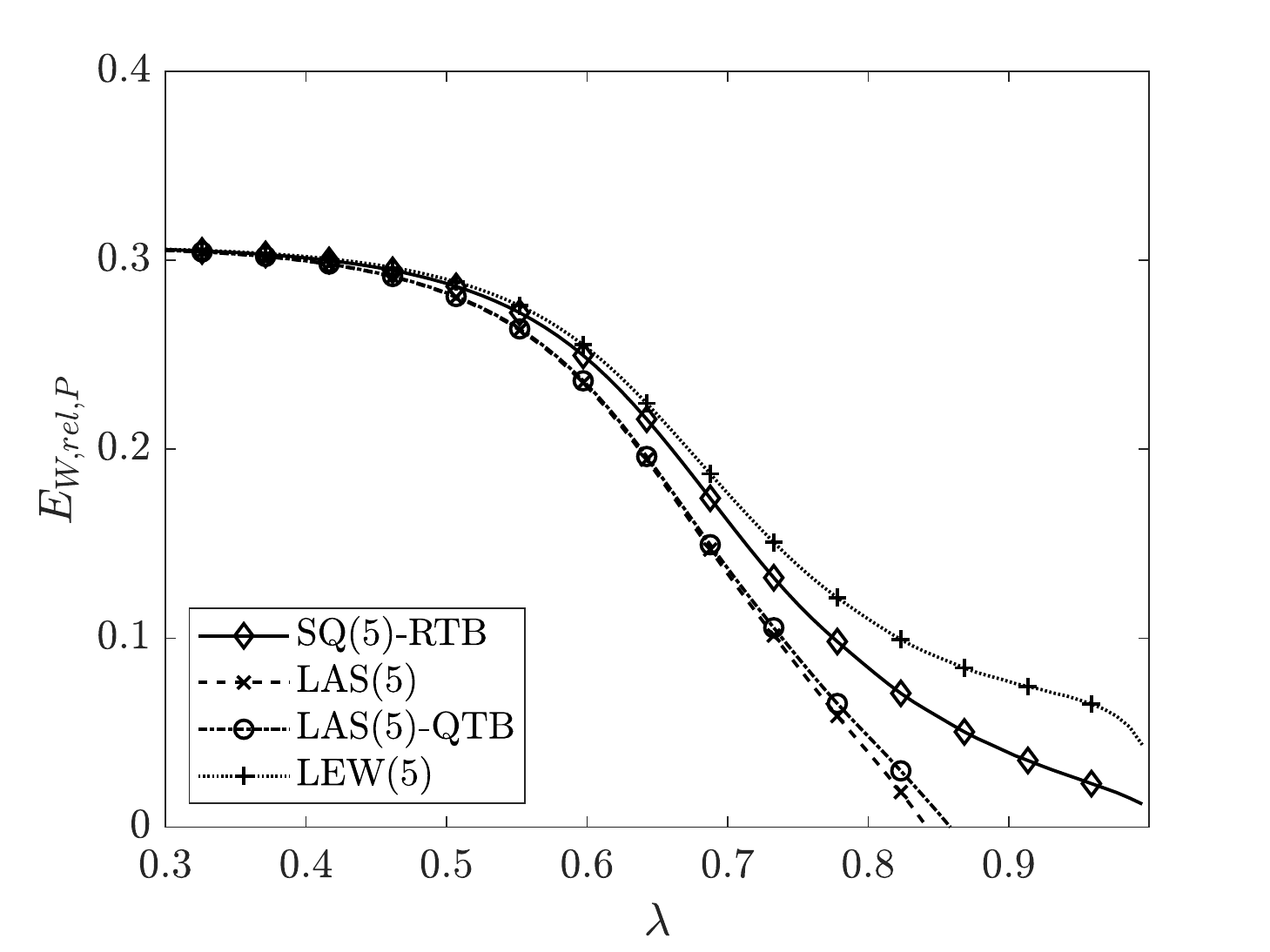}
\caption{Policies which do not make use of any threshold value.}
\label{fig8a}
\end{subfigure}
\begin{subfigure}{.4\textwidth}
\centering
\captionsetup{width=.8\linewidth}
\includegraphics[width=1\linewidth]{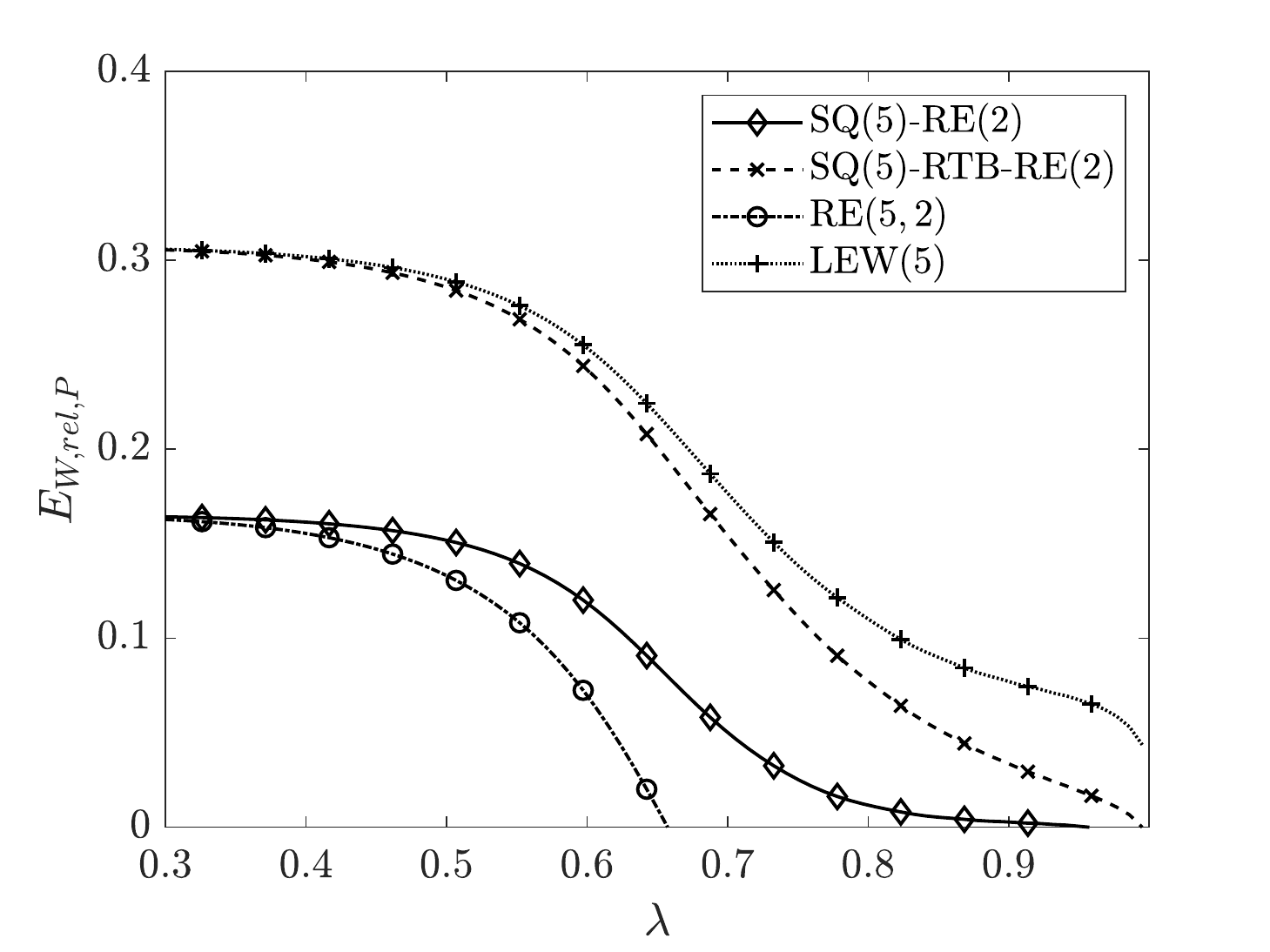}
\caption{Policies which make use of the threshold value $T=2$.}
\label{fig8b}
\end{subfigure}
\caption{Plots of the improvement in mean waiting time as a function of $\lambda$ for $d=5$, $\Delta=0.1$ and HEXP($10,\textbf{1/10}$) job sizes. This figure should be compared to Figure \ref{fig2}.}
\label{fig8}
\end{figure*}
In Figure \ref{fig8} we repeat the experiment in Figure \ref{fig2}, but we replace the value of $f$ by $f=1/10$. We observe that, somewhat surprisingly, having a small value for $f$ appears to decrease the gain from using the attained service time information. This may be explained by the fact that, as the value of $f$ decreases, the probability that all chosen servers are working on a large job increases. In particular for our example there is a probability around $(9/10)^5 \approx 59\%$ that all chosen servers are currently working on a large job (given that all chosen servers are currently busy). That is, in $59\%$ of cases, our policies reduce to either random routing or the standard SQ($d$) policy. We emphasize this point using Figure \ref{fig9}, where we show the performance of our policies as a function of $f$, we clearly see how the improvement reaches a peak around $f \approx 2/3$. The performance decreases sharply as $f\approx 0$ and $f \approx 1$, in the first case, (almost) all servers are working on large jobs, turning our policies into random routing resp.~SQ($d$) while for $f\approx 1$, almost all servers are working on short jobs, which again turns our policies into random routing resp.~SQ($d$).

\begin{figure*}[t]
\begin{subfigure}{0.4\textwidth}
\centering
\captionsetup{width=.8\linewidth}
\includegraphics[width=1\linewidth]{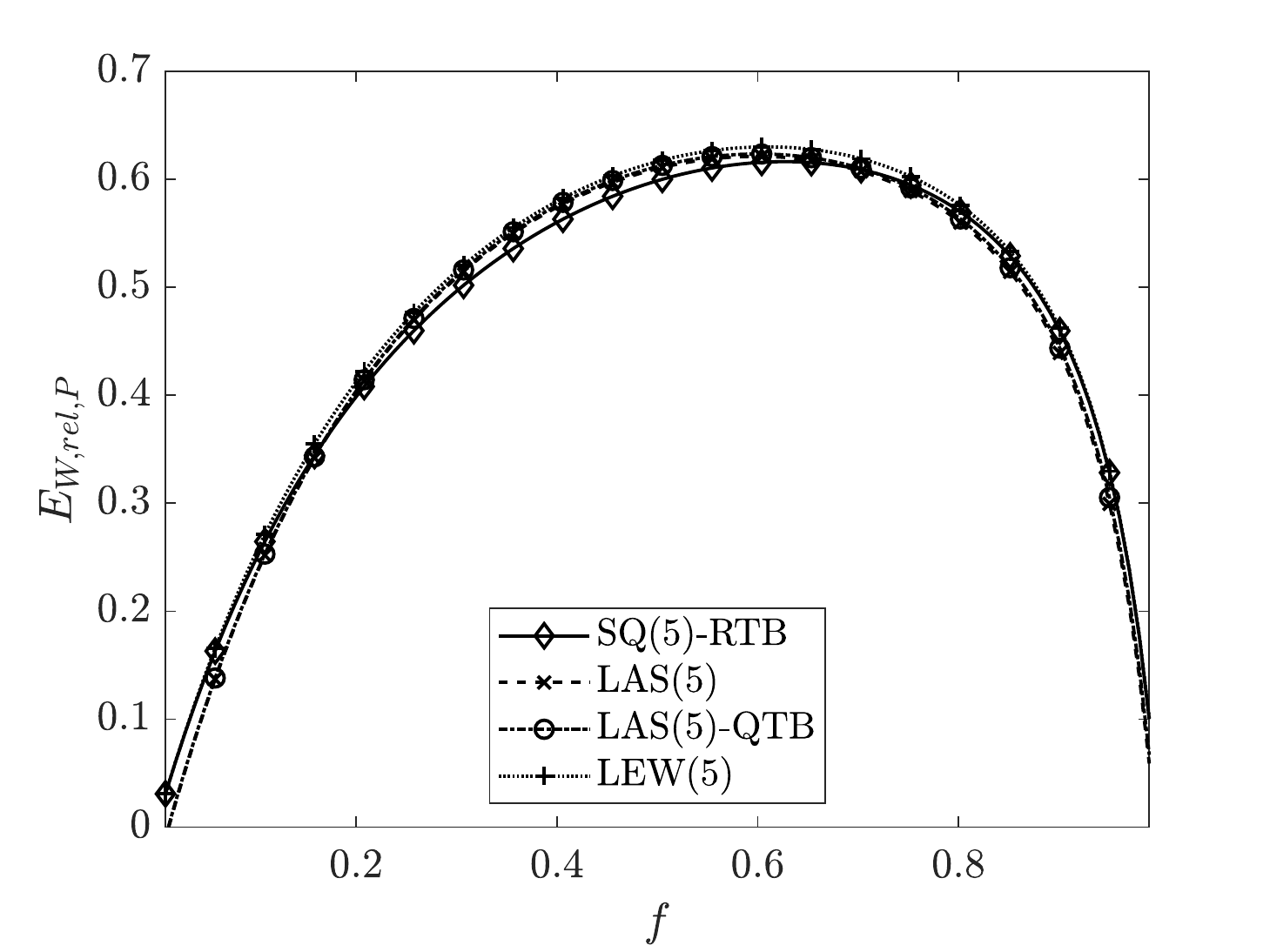}
\caption{Policies which do not make use of any threshold value.}
\label{fig9a}
\end{subfigure}
\begin{subfigure}{.4\textwidth}
\centering
\captionsetup{width=.8\linewidth}
\includegraphics[width=1\linewidth]{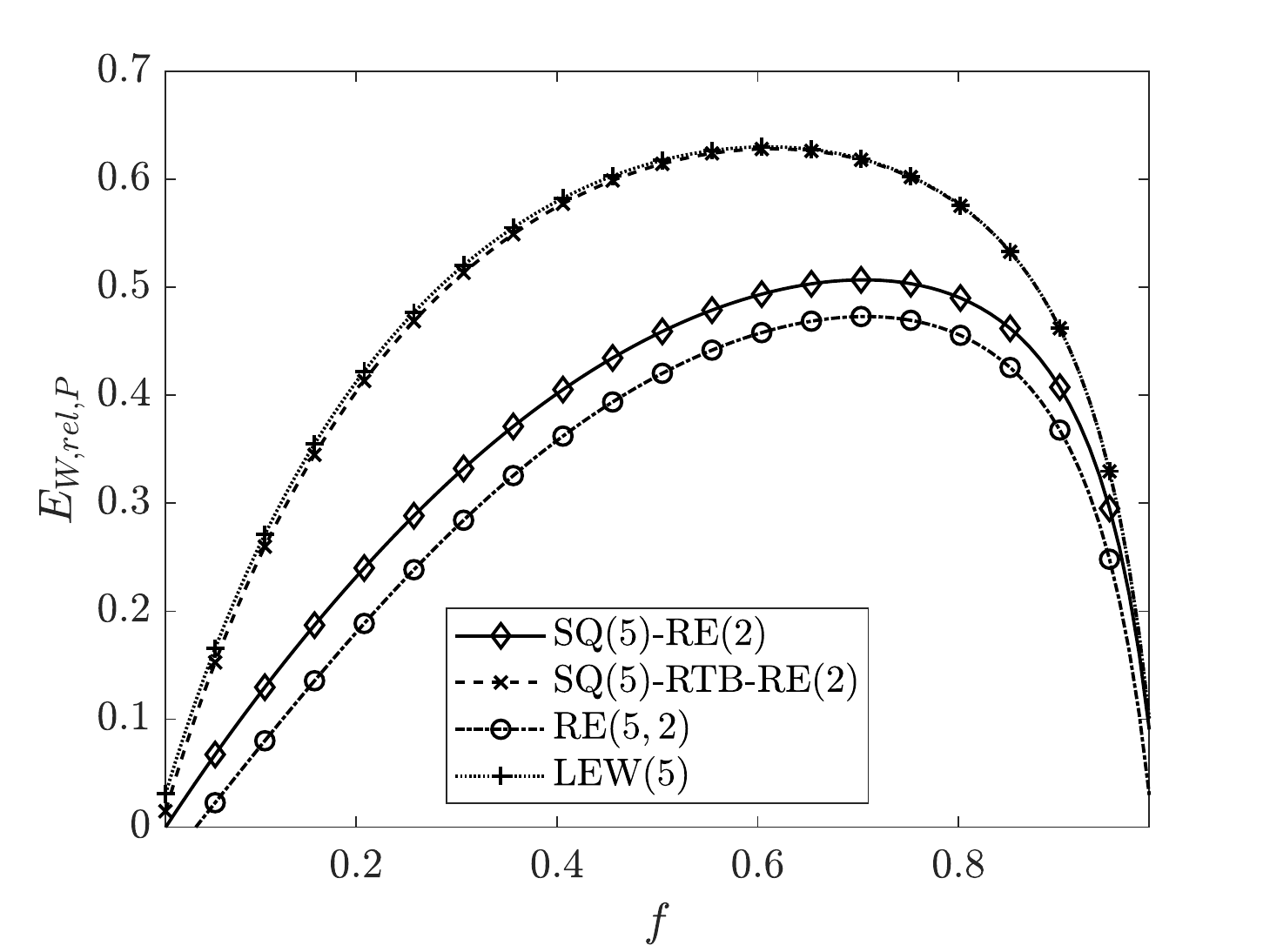}
\caption{Policies which make use of the threshold value $T=2$.}
\label{fig9b}
\end{subfigure}
\caption{Plots of the improvement in mean waiting time as a function of $f$ for $\lambda=0.6$, $d=5$, $\Delta=0.1$ and HEXP($10,\textbf{f}$) job sizes. This figure should be viewed as a supplement of Figure \ref{fig8}.}
\label{fig9}
\end{figure*}

\subsection{Tail of the Waiting Time Distribution}
\begin{figure*}[t]
\begin{subfigure}{0.4\textwidth}
\centering
\captionsetup{width=.8\linewidth}
\includegraphics[width=1\linewidth]{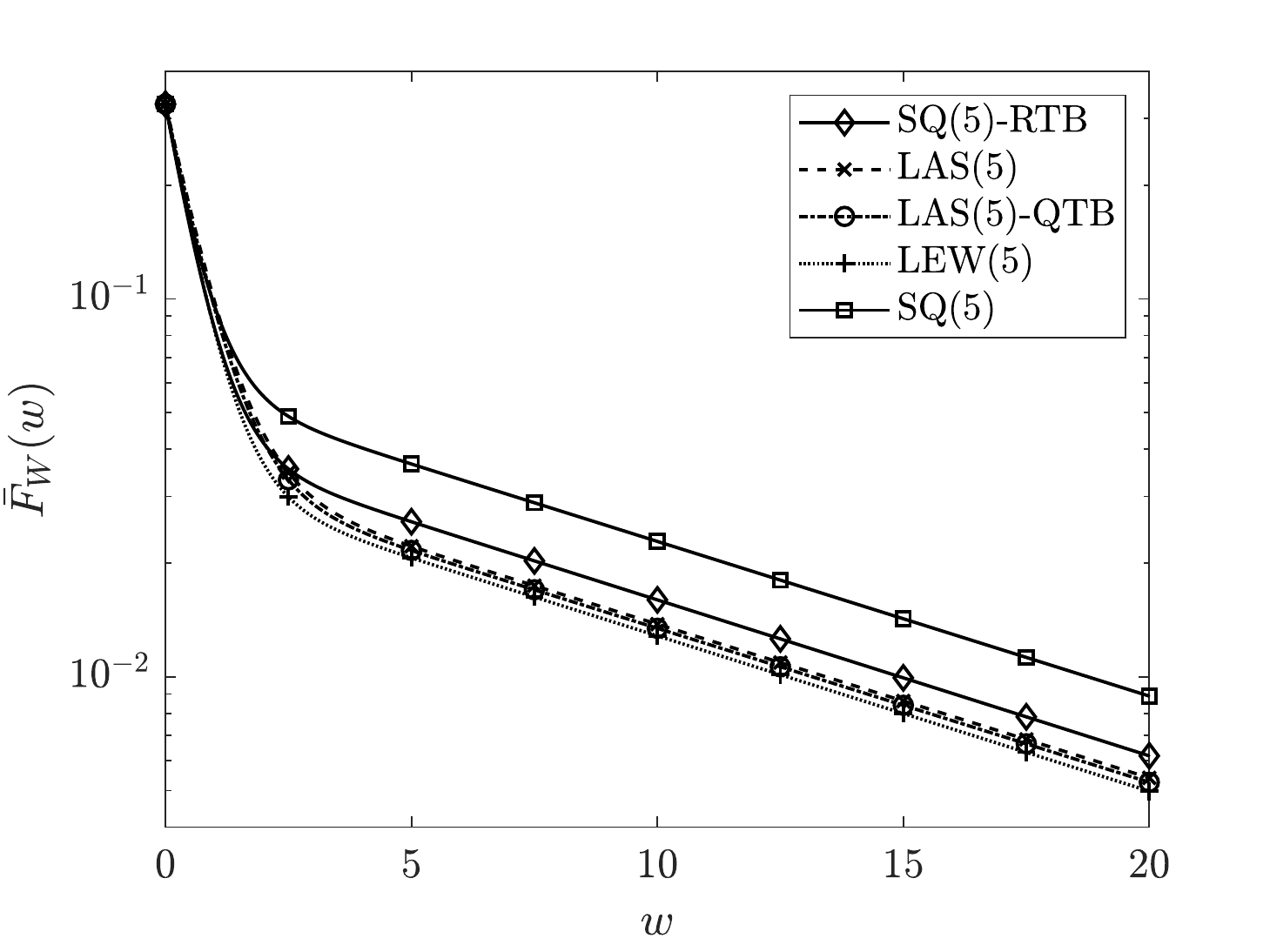}
\caption{Policies which do not make use of any threshold value.}
\label{fig10a}
\end{subfigure}
\begin{subfigure}{.4\textwidth}
\centering
\captionsetup{width=.8\linewidth}
\includegraphics[width=1\linewidth]{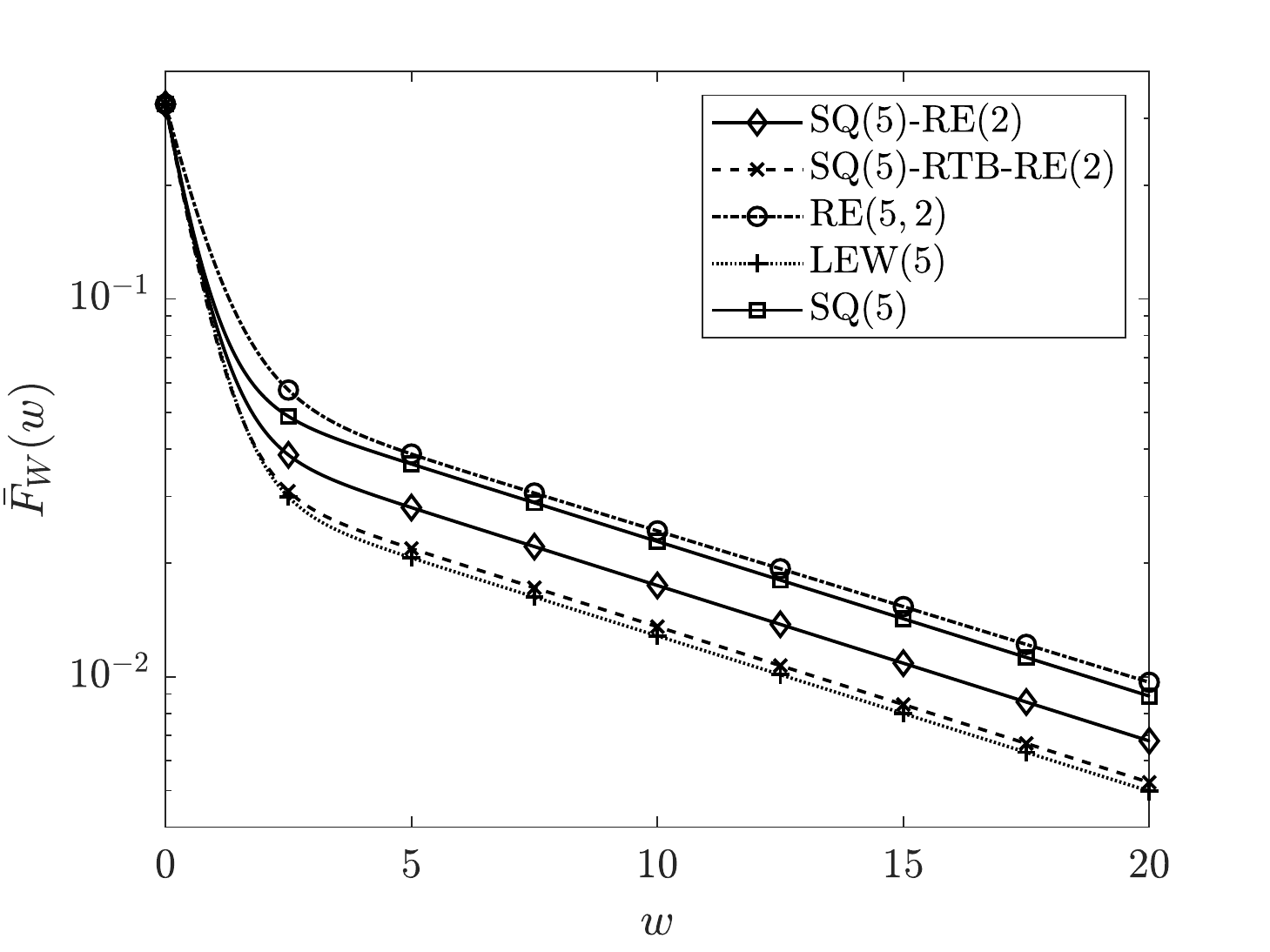}
\caption{Policies which make use of the threshold value $T=2$.}
\label{fig10b}
\end{subfigure}
\caption{Plots of the tails of the waiting time distributions $\bar F_W(w)=\mathbb{P}\{W \geq w\}$ for $\lambda=0.8$, $\Delta=0.1$, $d=5$ and HEXP($10,1/2$) job sizes. This figure should be seen as a supplement to Figure \ref{fig2}.}
\label{fig10}
\end{figure*}
In this section we take a closer look at the tail of the workload distribution $\bar F_W(w) = \mathbb{P}\{W \geq w\}$. In Figure \ref{fig10} we show $\bar F_W(w)$ as a function of $w$ for the same setting as in Figure \ref{fig2} with $\lambda = 0.8$ (i.e.~$d=5$, $T=2$ and HEXP($10,1/2$) job sizes). We observe that the tail behaviour is identical for all policies (and is the same as the tail of SQ($5$)). This result is not unexpected as the tail is heavily influenced by the
job size distribution of the long jobs.
This entails that studying the tail behaviour does not add much value to our discussion, therefore we keep our focus on the mean waiting time.
\subsection{The Granulity $\Delta$}
We have always used the same granularity $\Delta=0.1$, one could argue that in a real system the servers have more fine grained information on the attained service time of the job at the head of its queue. In Figure \ref{fig11a} we consider the same setting as in Figure \ref{fig2} ($d=5, T=2$ and HEXP($10,1/2$) job sizes), but with $\Delta=0.01$. We do not generate the plots associated to SQ($5$)-RE($2$) and RE($5, 2$) as these policies are independent of the granularity (the server simply states whether or not it has exceeded the threshold $T$). 
With the exception of LAS($5$)-QTB, we can hardly spot a difference with the plots in Figure \ref{fig2}. The exception for LAS($5$)-QTB can be explained by noting that
a smaller granularity $\Delta$ implies fewer ties in the reported $k$ value, meaning the queue length information
is neglected more often and as $\Delta$ tends to zero, the performance of LAS($5$)-QTB converges to that of LAS($5$) (while for a very large granularity its performance resembles SQ($d$)-RE($T$)). We confirm our findings in Figure \ref{fig11b}, where we repeat the plots made in Figure \ref{fig3}, but with $\Delta=0.01$ rather than $\Delta=0.1$ (i.e.~$d=5, T=2$ and HEXP($30, 1/2$) job sizes).
\begin{figure*}[t]
\begin{subfigure}{0.4\textwidth}
\centering
\captionsetup{width=.8\linewidth}
\includegraphics[width=1\linewidth]{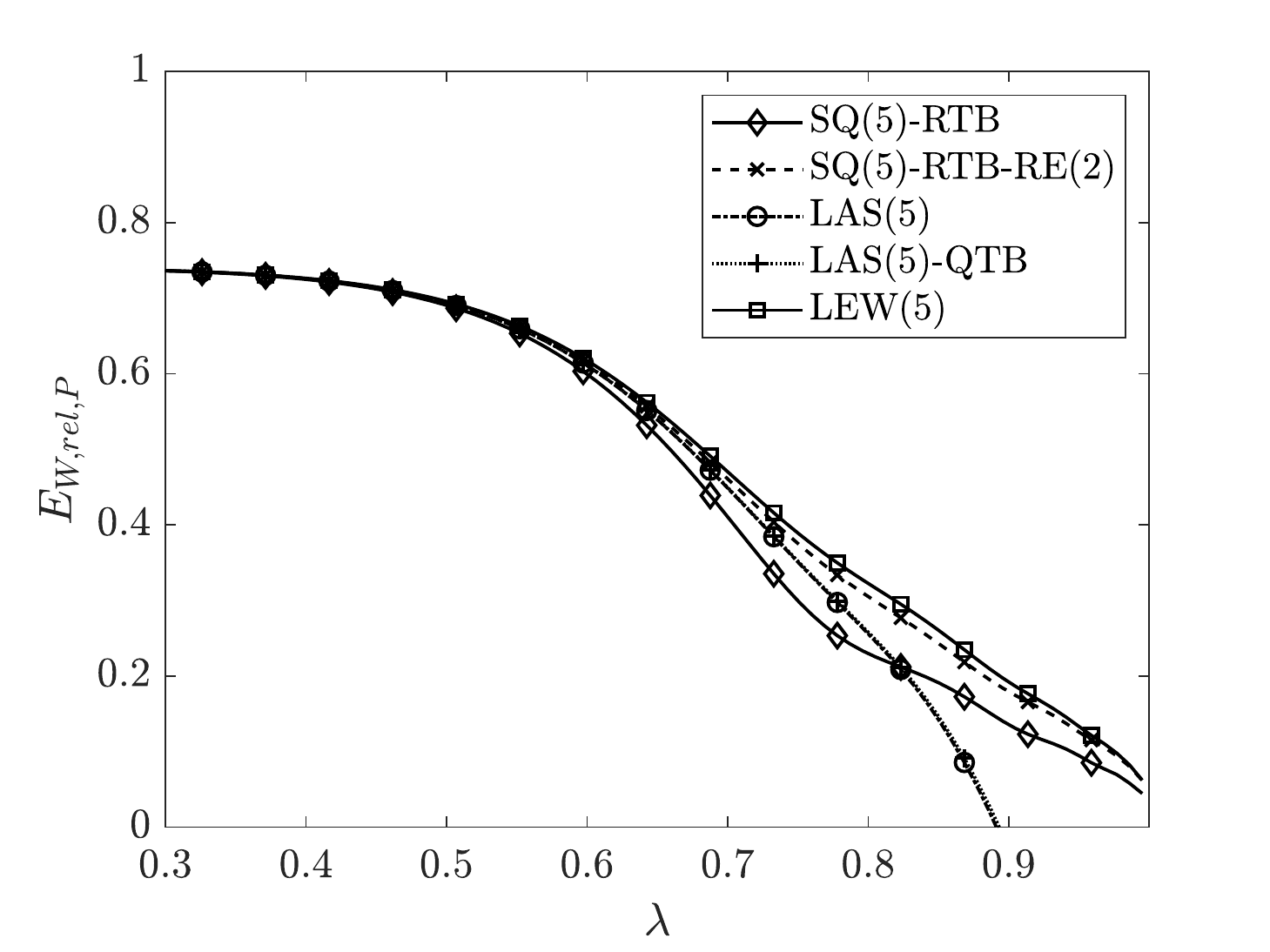}
\caption{For this plot, we set $SCV=10$, as such it should be compared to Figure \ref{fig2}.}
\label{fig11a}
\end{subfigure}
\begin{subfigure}{.4\textwidth}
\centering
\captionsetup{width=.8\linewidth}
\includegraphics[width=1\linewidth]{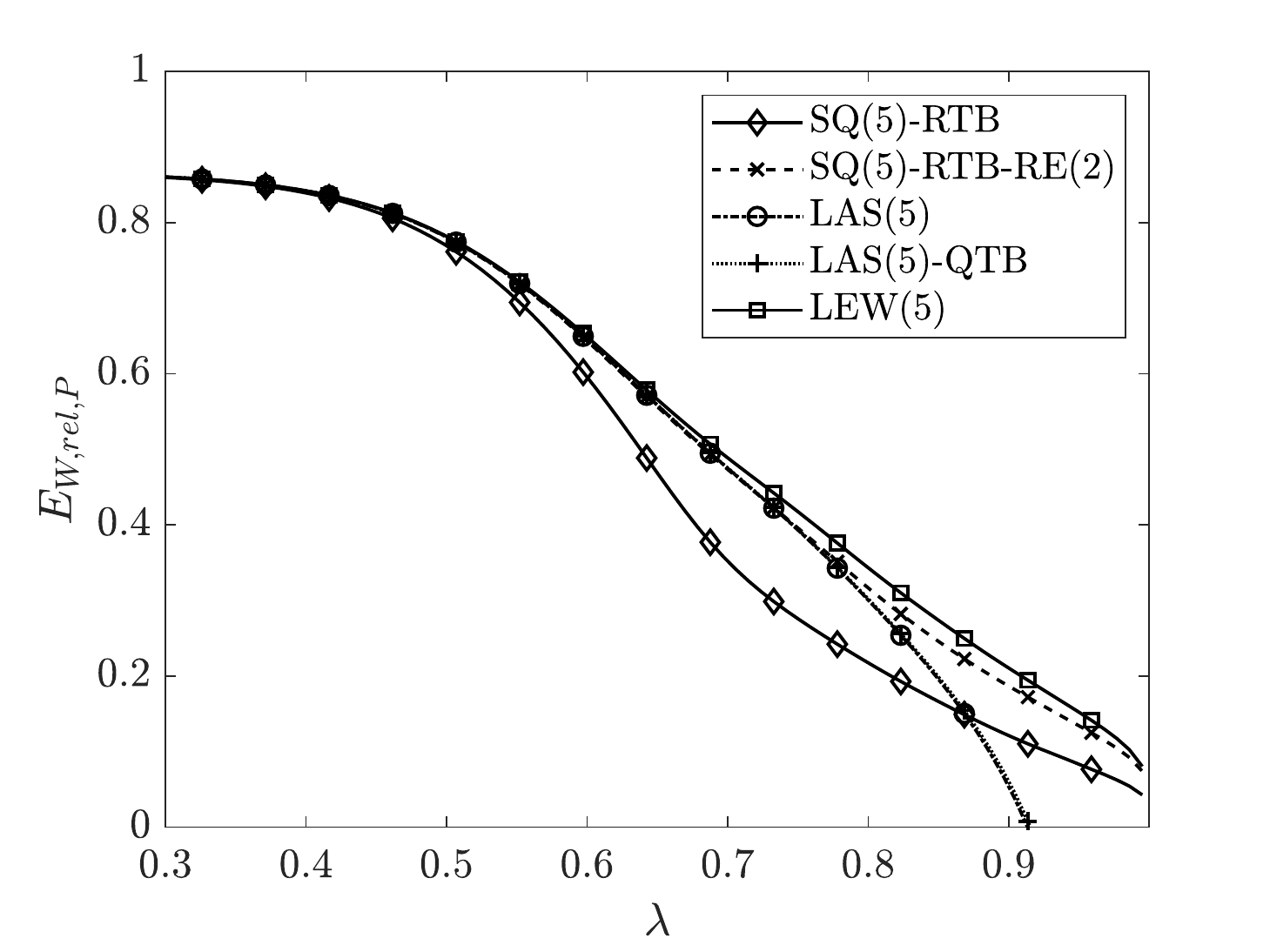}
\caption{For this plot, we set $SCV=30$, as such it should be compared to Figure \ref{fig3}.}
\label{fig11b}
\end{subfigure}
\caption{Plots of the improvement in mean waiting time as a function of $\lambda$ for $d=5$, $\boldsymbol \Delta$\textbf{=0.01} and HEXP($SCV, 1/2$) job sizes.}
\label{fig11}
\end{figure*}

\subsection{Choice of the threshold $T$} \label{sec:num_threshold}
\begin{figure*}[t]
\begin{subfigure}{0.4\textwidth}
\centering
\captionsetup{width=.8\linewidth}
\includegraphics[width=1\linewidth]{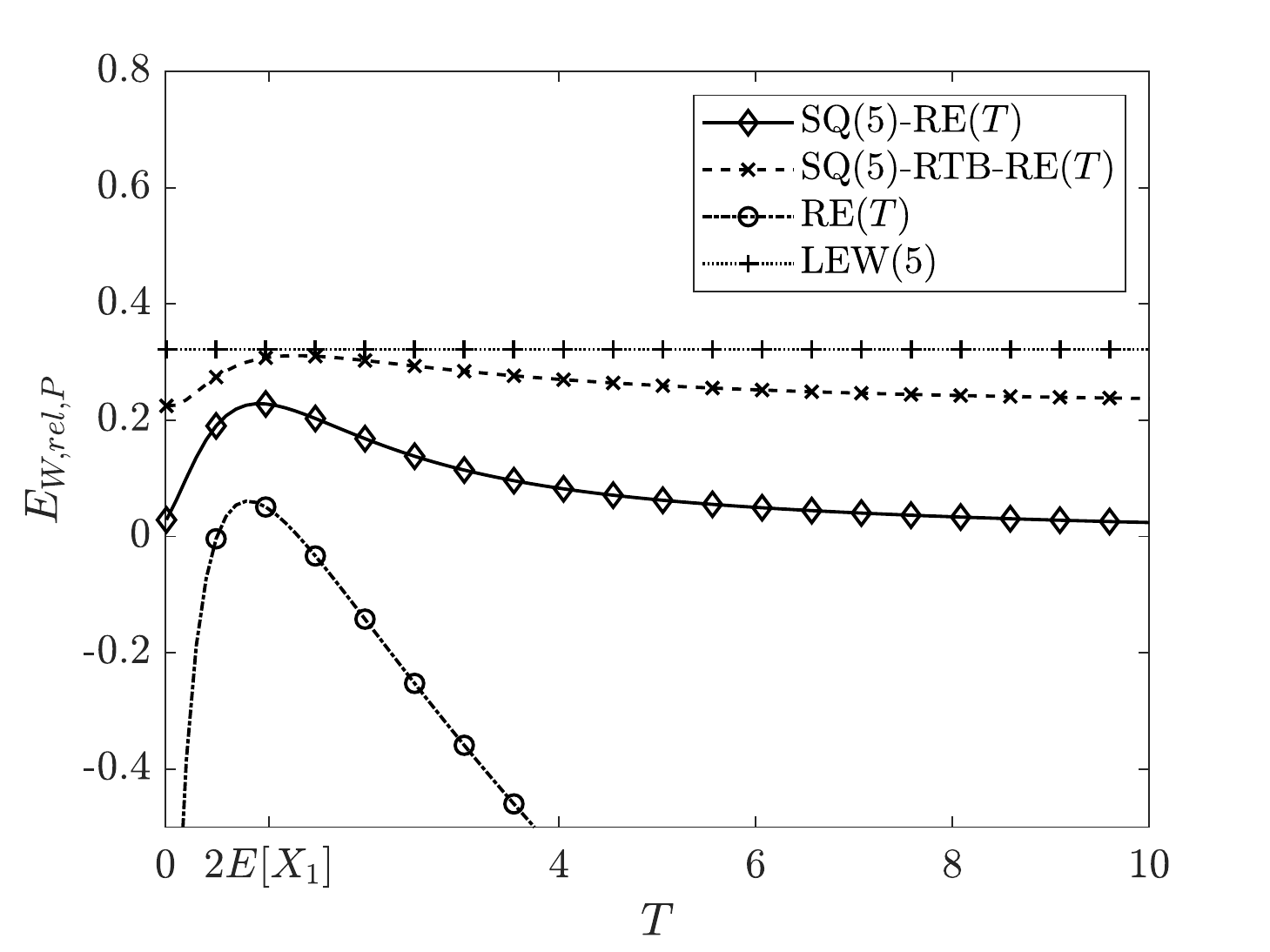}
\caption{For this plot, we set $f=1/2$, as such it should be compared to Figure \ref{fig2b}.}
\label{fig12a}
\end{subfigure}
\begin{subfigure}{.4\textwidth}
\centering
\captionsetup{width=.8\linewidth}
\includegraphics[width=1\linewidth]{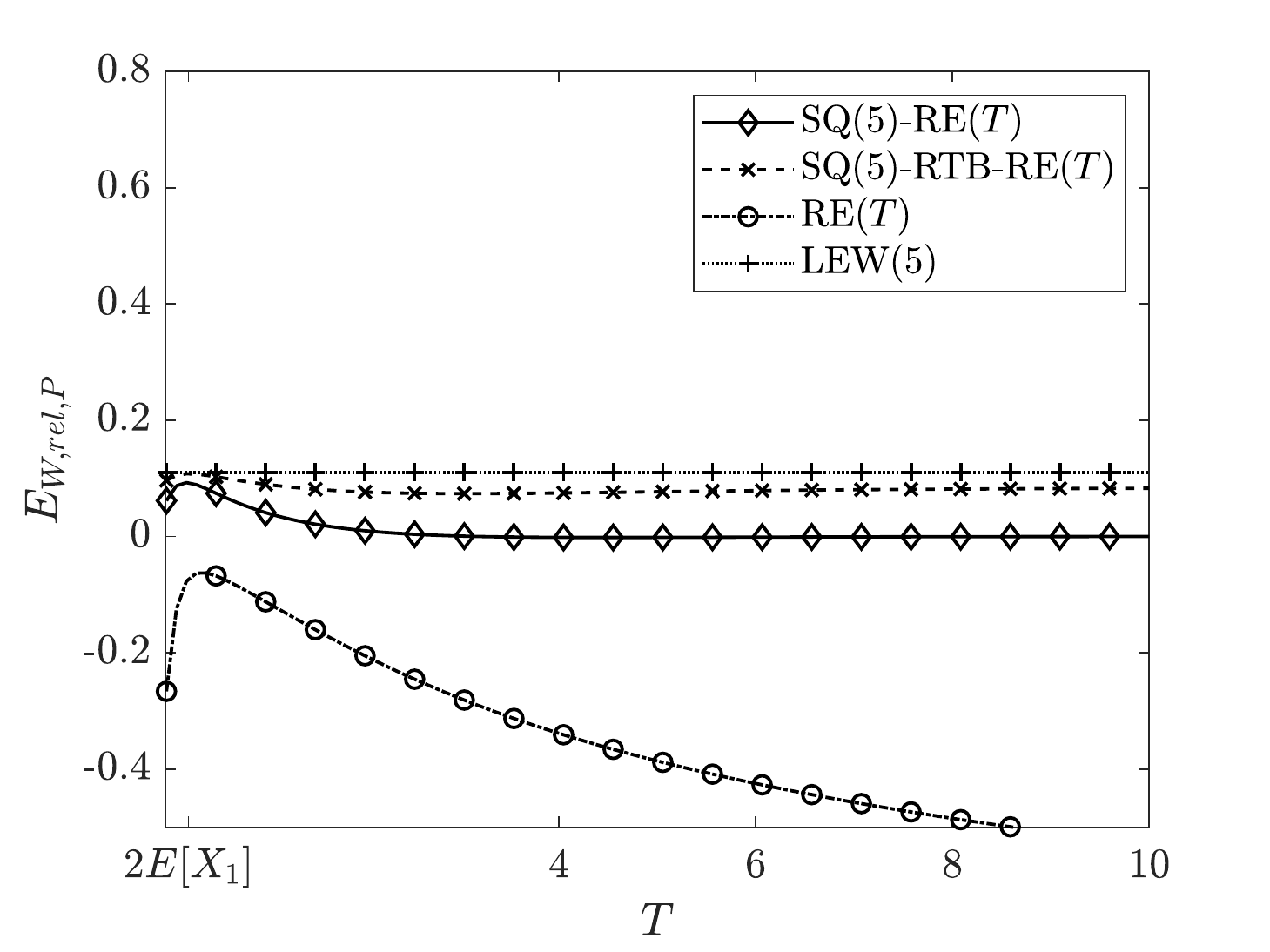}
\caption{For this plot, we set $f=1/10$, as such it should be compared to Figure \ref{fig8b}.}
\label{fig12b}
\end{subfigure}
\caption{Plots of the improvement in mean waiting time as a function of the threshold $T$ for $\lambda=0.8$, $d=5$, $\Delta=0.1$ and HEXP($10,f$) job sizes. We denote the job size distribution of the small jobs by $X_1$.}
\label{fig12}
\end{figure*}
Thus far, we have always used a (somewhat arbitrary) threshold $T = 2 \cdot \E[X]$. With this choice we noticed that we obtain a significant improvement over the standard SQ($d$) policy. However, it would make more sense if we chose a threshold which depends on the mean of the small jobs rather than the mean of all jobs. This way, our threshold is more directly linked to the probability that a job is long given that a server has been working on it for a time $T$. We now look at the impact of $T$ on the performance of the policies which depend on a threshold value. To that end, we use the setting of Figures \ref{fig2b} and \ref{fig8b} with $\lambda=0.8$ (that is, we take $d=5$ and consider HEXP($10, f$) with $f$ equal to $1/2$ and $1/10$). Let $X_1$ denote the job size distribution for the small jobs, then we have $\E[X_1] \approx 0.53$ in the case of $f=1/2$, while $\E[X_1] \approx 0.12$ for $f=1/10$. In Figure \ref{fig12} we notice that the plots attain a maximum and this maximum appears to be close to the value $T \approx 2 \E[X_1]$.

In Figure \ref{fig13} we use this improved threshold value of $T=2\E[X_1]$ and show the improvement in mean waiting time as a function of the arrival rate $\lambda$. We observe that (comparing with Figures \ref{fig2b} and \ref{fig8b}) there is a noticeable improvement in performance. This is especially true for the case of $f=1/10$, which is quite natural as for this case we now take $T=0.24$ instead of $T=2$ while for $f=1/2$ the threshold value of $2$ is replaced by the value $T=1.06 > 0.24$. However, we note that in both cases, the improvement was still significant for all considered policies even when we were using a suboptimal threshold value.
\begin{figure*}[t]
\begin{subfigure}{0.4\textwidth}
\centering
\captionsetup{width=.8\linewidth}
\includegraphics[width=1\linewidth]{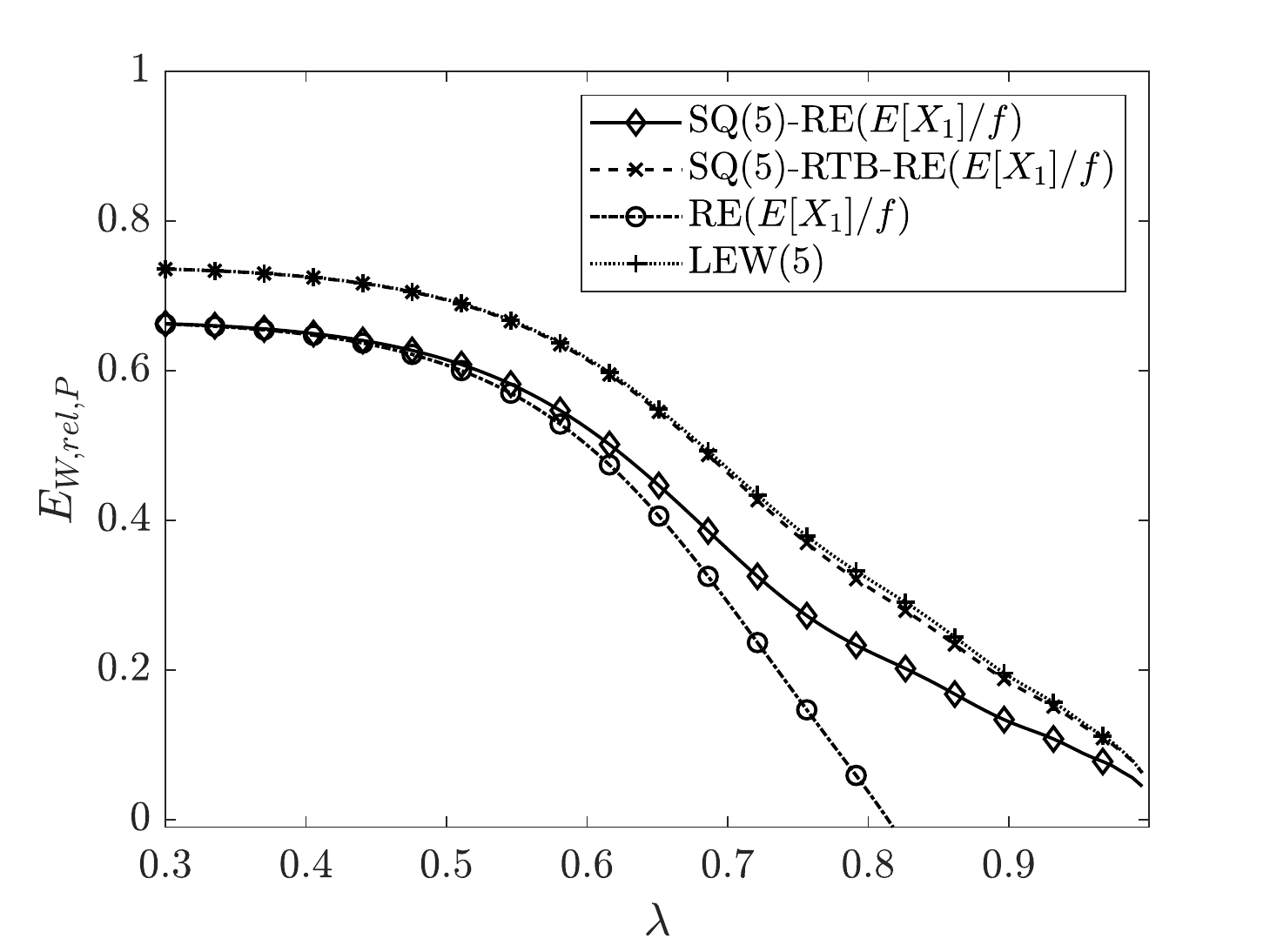}
\caption{Plot for $f=1/2$, this figure should be compared to Figure \ref{fig2b}.}
\label{fig13a}
\end{subfigure}
\begin{subfigure}{.4\textwidth}
\centering
\captionsetup{width=.8\linewidth}
\includegraphics[width=1\linewidth]{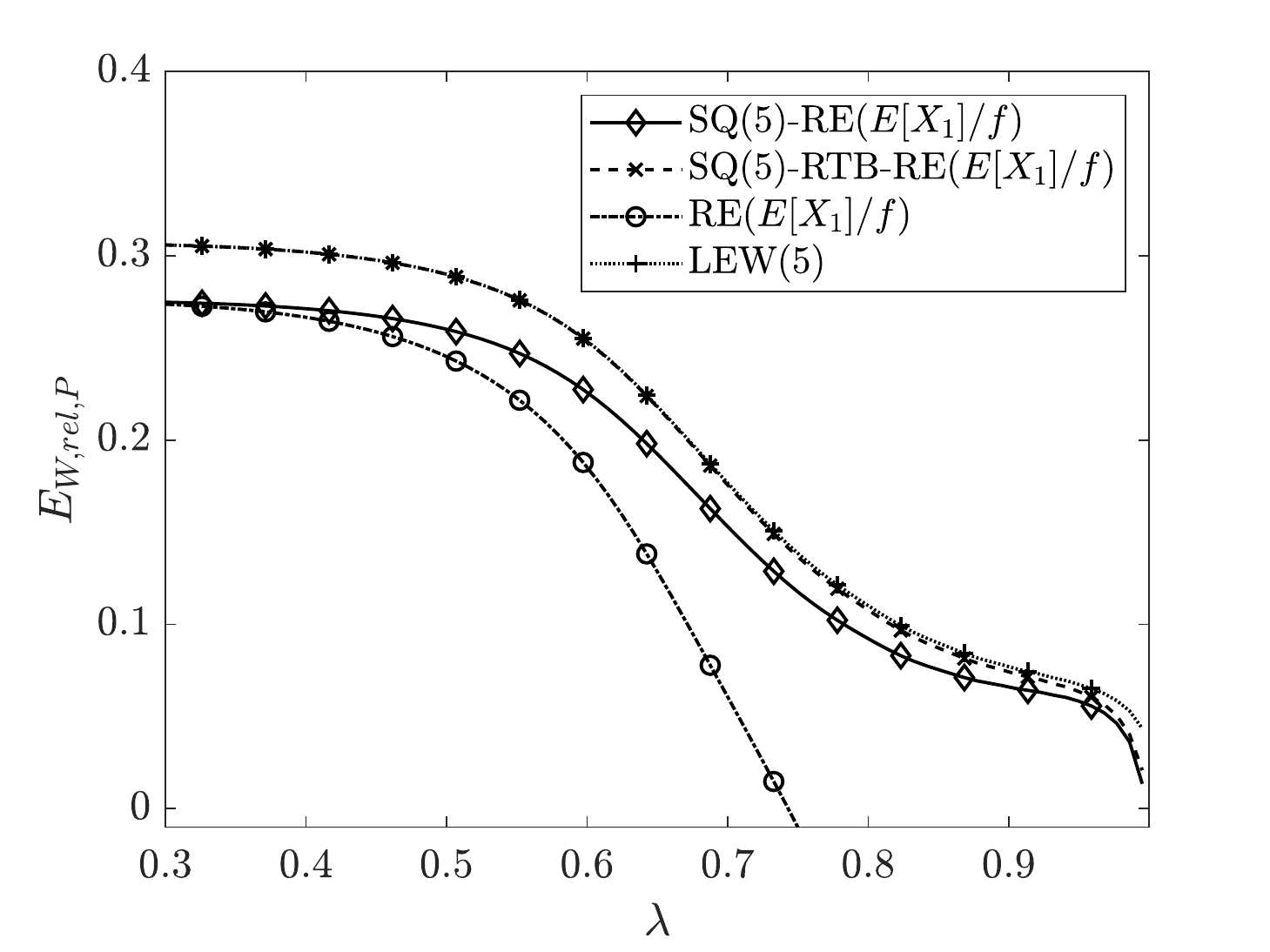}
\caption{Plot for $f=1/10$, this figure should be compared to Figure \ref{fig8b}.}
\label{fig13b}
\end{subfigure}
\caption{Plots of the improvement in mean waiting time as a function of $\lambda$ for $d=5$, $\Delta=0.1$, HEXP($10,f$) job sizes and using the improved threshold value $T = 2\cdot \E[X_1]$, with $X_1$ the job size distribution of the small jobs.}
\label{fig13}
\end{figure*}

\section{Conclusions and Future Work} \label{sec:conclusion}
In this paper we showed that the mean waiting time of the SQ($d$) policy can be significantly reduced using simple policies if the servers report the attained service time in addition to
their queue length when the workload consists of a mixture of short and long jobs. 
The attained service time is a quantity that is easy to measure in
FCFS servers.  It is even possible to achieve a sizeable improvement over SQ($d$) by simply setting a single threshold value $T$ such that a server can flag a job as large when this threshold is exceeded. While choosing a good value for $T$ may further improve performance, most policies are not too sensitive to $T$ and there is a wide range of $T$ values which 
achieve a notable performance improvement.

We did note that the improvement coming from the attained service time information 
generally decreases as the system load approaches one, as the attained service time
information becomes less valuable in the presence of long queues. However, if one were to use another scheduling policy which is aware of the attained service time of multiple jobs in its queue, larger gains may also be feasible at high load (e.g.~using Processor Sharing).

In our experiments we have focused solely on job sizes which represent a system in which both large and small jobs arrive. We 
believe however that the insights obtained in our numerical experiments extend far beyond this
class of distributions. Further, the approach developed in the paper to assess the system performance can be used for any phase-type distribution (including fittings of heavy tailed distributions).

The main idea presented in this paper can also be used in systems where there is
some form of memory at the dispatcher. For instance, the approach of
\cite{van2019hyper} could be adapted such that the dispatcher 
maintains both an upper bound on the number of jobs in each queue and a lower bound on the 
attained service time of the job at the head of each queue.

\bibliographystyle{ACM-Reference-Format}
\bibliography{thesis}

\end{document}